\documentclass[letterpaper, 11pt, one column]{article}
\usepackage[margin=1in]{geometry}
\usepackage[labeled,resetlabels]{multibib}
\usepackage{hyperref}
\hypersetup{
  colorlinks = true,
  linkcolor=blue,   
  citecolor=blue,   
  urlcolor=blue,    
  pagebackref=true,
  implicit=false,
  bookmarks=true,
  bookmarksopen=true,
  pdfdisplaydoctitle=true
}

\AtBeginDocument{%
  }

\usepackage{standalone}
\usepackage{amsmath}
\usepackage{amsfonts}
\usepackage{mathrsfs}
\usepackage{amsthm}
\usepackage{mathtools}
\usepackage{algpseudocode}
\usepackage{algorithm}
\usepackage{tikz}
\usetikzlibrary{decorations.text, decorations.pathreplacing}
\usetikzlibrary{positioning, decorations.markings, calc}
\usetikzlibrary{decorations.pathreplacing}
\usepackage{mathrsfs}

\usepackage[mathscr]{euscript}
\usepackage{authblk}
\usepackage{natbib}
\setcitestyle{authoryear} 

\newcommand{\cost}{\operatorname{cost}}
\newcommand{\OPT}{\operatorname{OPT}}
\newcommand{\opt}{\operatorname{OPT}}

\theoremstyle{plain}
\newtheorem{thm}{Theorem}[section]
\newtheorem{lem}[thm]{Lemma}
\newtheorem{prop}[thm]{Proposition}
\newtheorem{cor}[thm]{Corollary}

\theoremstyle{definition}
\newtheorem{defn}{Definition}[section]

\theoremstyle{remark}

\begin{document}

\title{Low Cost, Fair, and Representative Committees in a Metric Space}

\author{Christopher Jerrett}
\author{Elliot Anshelevich}

\affil{Department of Computer Science, Rensselaer Polytechnic Institute}
\affil{jerrec@rpi.edu, eanshel@cs.rpi.edu}


\maketitle

\begin{abstract}
We study the problem of selecting a representative committee of $k$ agents from a collection of $n$ agents in a common metric space. This problem is related to choosing $k$ facilities in facility location and $k$-median problems. However, unlike in more traditional facility location where each agent only cares about the {\em closest} selected facility, in the settings we consider each agent desires that {\em all} selected committee members are close to them. More precisely, we look at the sum objective, in which the goal is to minimize the total distance from all agents to all members of the chosen committee. We show that it is always possible to find a committee which is both {\em low-cost} according to this objective, and also {\em fair} according to many existing notions of fairness and proportionality defined for clustering settings. Moreover, we introduce a new desirable axiom for representative committees we call {\em NORP}, which prevents {\em over-representation} of any subset of agents. While all existing algorithms for fair committee selection do not satisfy this intuitive property, we provide new algorithms which form simultaneously {\em low-cost, fair, and NORP} solutions, thus showing that it is always possible to form low-cost, fair, and representative committees for our settings. 
\end{abstract}

\section{Introduction}
Selecting a small set of representatives from a large set of points in a metric space is an important and foundational problem in AI and social choice.
The distances in the metric space may encode, for example: physical proximity (e.g., for choosing a set of facilities in facility location), ideological alignment (e.g., political spectra and spatial voting \cite{enelow1990advances, merrill1999unified}), or general dissimilarity (e.g., dissimilarity between images when attempting to select a small set of images which are representative of a large collection \cite{kalayci2024proportional,sun2012optimizing}).
Consider forming a committee: there are $n$ agents (voters) embedded in some ideological metric space, and we want to choose a small number of these voters as our committee of representatives.
Crucially, the agents care about all of the chosen representatives' ideologies (as modeled by their locations in the metric space), not just their closest one.
If all representatives have equal voting power, then agents should care about how much they like each chosen representative equally. 
We model this by attempting to minimize the total distance from all agents to all chosen representatives. 

Formally, in our problem we are given a set $N$ of $n$ agents, an integer $k \leq n$ and a (pseudo)metric\footnote{We consider pseudometrics since we allow multiple different agents to be located at the same point, i.e., have distance 0 between them.} $d:N \times N \to \mathbb{R}^+$. We wish to select a $k$ element subset $X\subseteq N$ which we will call the ``representatives" of $N$. We consider $X$ to be a good solution if it minimizes the natural sum cost function discussed above, i.e., $\sum_{v \in N} \sum_{x \in X} d(v,x)$. Note that, unlike other typical clustering objectives (including $k$-median, $k$-means, and $k$-center), efficiently computing a solution minimizing our cost function is trivial using a simple greedy algorithm. 

The sum objective we consider here has been extensively studied in the heterogeneous facility location literature (see e.g., \cite{anastasiadis2018heterogeneous, deligkas2023heterogeneous,serafino2014heterogeneous,kanellopoulos2023discrete}). In this setting, we wish to construct $k$ facilities in a metric space that minimize the sum objective. For example, consider the situation where a city planner must decide where to construct a collection of heterogeneous public facilities (a school, a post office, a court house, a city hall, a library, and a park). Importantly, {\em each facility serves a different purpose}, meaning that agents will need to visit and use all the facilities we build rather than rely on one or a few of them. This is distinct from clustering or more standard facility location problems, where agents are assigned to just one (typically the closest) facility. See Related Work for further discussion on this topic. 

\paragraph{Fair and Low-cost Solutions}
As discussed above, minimizing the sum objective often makes sense for committee selection, heterogeneous facility location, and other settings where the goal is to select a small set of representative agents with each agent caring about their distance to {\em all} of the selected representatives. Thus, for the settings we are interested in, why not simply choose the solution optimizing this cost function, which we will call the {\em min-cost solution}? Unfortunately, while such a solution has clear advantages, and is easy to compute efficiently, it also has some important downsides. 

Perhaps the main disadvantage of the min-cost solution is that it can be very {\em unfair}.
When deciding on a committee, we often want different groups to have representation proportional to their size \cite{pukelsheim2017proportional}, even if it hurts the sum objective. 
Consider a very simple example on the line segment $[0,1]$ where there are $51\%$ of agents located at 0 and $49\%$ located at 1. It seems reasonable, and consistent with basic principles of proportionality \cite{aziz2020expanding, brill2023robust}, to have about half of the representatives at 0 and half at 1. The min-cost solution, however, does not do this: all representatives are chosen to be at 0. Intuitively, this seems unfair; the agents located at 1 deserve to have at least one representative nearby, and preferably about $\frac{k}{2}$ representatives.
This reveals a tension: the sum objective optimizes global efficiency but may ignore local fairness.\footnote{More generally, the min-cost solution may not be fair or proportional for any reasonable notion of proportionality, including PRF, mJR, $\alpha$-$1$-core, mPJR, NORP, etc (see Related Work for discussion of these fairness concepts).} 
Yet in many applications—from political representation to public facilities to diverse image selection—both are essential. 
We thus would like to include some notion of fairness and proportionality in our idea of a desirable solution.  Our main goal in this paper is to show how to form simultaneously \textbf{fair and low-cost solutions}. By ``fair" we mean one that satisfies certain fairness axioms, discussed below, and by ``low-cost" we mean one that is a constant factor approximation of the optimal solution in the sum objective.\footnote{Note that it is {\em not} possible to do this for many objectives, including $k$-median, $k$-means, and $k$-center (see Related Work). Unlike these clustering objectives, the sum objective permits both fairness and constant-factor approximations {\em simultaneously}.}

\paragraph{Proportionally Representative Fairness} Many fairness properties have been introduced for our setting during the last few years (see detailed discussion in Related Work). We begin by focusing on Proportional Representative Fairness (PRF), introduced by \cite{aziz2023proportionallyrepresentativeclustering}. Suppose that we have $\ell \frac{n}{k}$ agents located at a single point. It is reasonable to require that at least $\ell$ of the $k$ representatives should also be located at that point, to ensure proportional representation, since these agents form $\frac{\ell}{k}$ fraction of the total. More generally, if a group of $\ell \frac{n}{k}$ agents is {\em not} located at a single point, it still makes sense to require that there be at least $\ell$ representatives nearby. For a cohesive group where all agents are near each other (a densely populated city), ``nearby" should mean quite close, while for a less cohesive group where the agents may be far away from each other (a sparse collection of small rural villages), it would make sense for ``nearby" to mean farther as well. In particular, if we define $D(S)= \max_{v,w \in S} d(v,w)$ to be the diameter of a set $S$ of agents, then ``nearby" can mean ``proportional to the diameter". This brings us to the following formal definition, which simply states that for any large set $S$, there are at least $\ell$ representatives within distance at most $D(S)$ from the set $S$. 
\begin{defn}
    A representative set $X$ of size $k$ satisfies {\em Proportionally Representative Fairness} (PRF) if the following holds. For any set of agents $S \subseteq N$ of size $\vert S \vert \geq \ell \frac{n}{k}$, it must be that $\vert x \in X:\exists v \in S \text{ s.t. } d(v,x) \leq D(S) \vert \geq \ell$. 
\end{defn}

While we study other fairness concepts as well, we choose to focus on PRF for several reasons: (1) PRF solutions always exist, and can be computed efficiently.\footnote{We study the setting in which the set of agents is the same as the set of possible representatives. In \cite{aziz2023proportionallyrepresentativeclustering} they show that if the entire metric space is the candidate set then PRF exists; however they do this by only considering the locations of the agents as candidates, so their results apply to our setting as well.} 
This is in contrast to many other fairness concepts, for which only approximately fair solutions are guaranteed to exist.
(2) We believe PRF neatly captures the basic intuition behind fairness and proportionality: a large set should have many representatives near it, where ``near" means proportional to how spread out the set is. (3) The previous reason is, of course, subjective, and the reader may prefer one of the other known fairness concepts, described in Related Work. However, most of our results also hold for other fairness concepts in addition to PRF. In particular, we show that almost all of our results also hold for Metric Justified Representation (mJR, see Section \ref{sec:model} for definition), which due to \cite{kellerhals2023proportional} immediately implies results for Proportional Fairness, Transferable Core, and Individual Fairness. Thus, if the reader prefers a different fairness notion to PRF, using \cite{kellerhals2023proportional} our results for PRF and mJR likely extend to that as well.

\paragraph{No Over-Representation Property (NORP)}
Our work consists of two main contributions. First, we show that the sum objective is compatible with many existing notions of fairness, including PRF and mJR, and give efficient algorithms to compute simultaneously low-cost and fair solutions. Second, we introduce a novel fairness axiom to address the problem of over-representation, and show that it is possible to satisfy this property without losing much, i.e., while still forming fair and low-cost solutions. 



An undesirable feature of many existing fairness axioms is that they can allow a small group of agents to be over-represented. Consider the following simple example, with large $n$ and $k = 3$ given in Figure \ref{ex:star}. Usual fairness axioms would allow all three representatives to be placed in the center. From the point of view of the $n-3$ agents which are far away from the center, however, the 3 agents in the center are extremely privileged: why do those 3 agents get all 3 representatives right next to them? Moreover, the three chosen representatives lack diversity, as they are all at essentially the same point. If we instead place one representative in the center and two in the outer groups, then fairness axioms are still preserved while solving the above concern about the center location being over-represented.

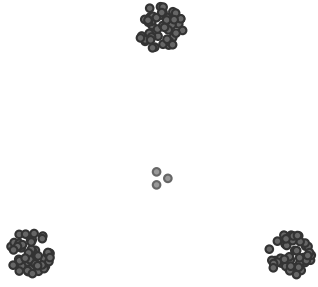
\begin{figure}[h]

    \centering
\begin{tikzpicture}

\def\n{150}  

\def\radius{2}  
\def\agentradius{0.05}  
\def\clusterradius{0.3}  

\foreach \i in {0,1,2} {
    \pgfmathsetmacro\angle{120*\i}
    \pgfmathsetmacro\cx{0.1*cos(\angle)}
    \pgfmathsetmacro\cy{0.1*sin(\angle)}
    \filldraw[black!40, draw=black!60, thick] (\cx,\cy) circle (\agentradius);
}

\pgfmathtruncatemacro\agentspergroup{(\n-3)/3}

\foreach \groupnum in {0,1,2} {
    \pgfmathsetmacro\groupangle{120*\groupnum + 90}  
    
    \pgfmathsetmacro\gcx{\radius*cos(\groupangle)}
    \pgfmathsetmacro\gcy{\radius*sin(\groupangle)}
    
    \pgfmathsetseed{42 + \groupnum*1000}  
    \foreach \i in {1,2,...,\agentspergroup} {
        \pgfmathsetmacro\randangle{rnd*360}
        \pgfmathsetmacro\randdist{\clusterradius*sqrt(rnd)}  
        \pgfmathsetmacro\ax{\gcx + \randdist*cos(\randangle)}
        \pgfmathsetmacro\ay{\gcy + \randdist*sin(\randangle)}
        \filldraw[black!60, draw=black!80, thick] (\ax,\ay) circle (\agentradius);
    }
}

\end{tikzpicture}
    \caption{An example with a large number of agents in the plane and $k=3$. We have three large groups of $\frac{n}{k}-1$ agents placed symmetrically around the origin, with $3$ agents at the origin. Then standard fairness axioms would allow all three representatives to be placed in the center.}
    \label{ex:star}
\end{figure}

The example in Figure \ref{ex:star} can be easily generalized to arbitrary $k$ by using a star graph topology with $k$ agents in the center and $\frac{n}{k}-1$ agents at each of the arms of the star (see Figure \ref{fig:notnorpalg2}). To the best of our knowledge, all fairness notions existing in the literature so far would allow all or nearly all representatives to be placed in the center in such examples. All existing algorithms developed for our setting would also place all or nearly all representatives on top of the $k$ center agents, giving those $k$ agents disproportionate power, while also creating envy from almost all other agents.
Why does this very small group of $k$ agents in the middle have almost all the representatives directly on top of them? This is an undesirable property from a few different perspectives. The first is that agents could be envious of the small group of ``elite" agents where representation is concentrated. In real settings, people are not only concerned with having proportionally enough, but also resent if someone else has disproportionally too much. Second, by allowing too many representatives to be concentrated in a small region, the solution could become homogeneous and lack diversity, making it less robust. 
Thus, we want a fairness concept that says that if a small area in the space has a large number of representatives, there should be a proportional number of agents nearby, so that area is truly deserving of so many representatives. 
To form this fairness concept, we consider arbitrary sets $S$ of $\ell$ representatives. How many agents should be near $S$ to ``deserve" $S$ being there, and what does ``near" mean? We say that there are enough agents nearby if there are enough to justify $\ell-1$ representatives, along with at least one more agent to justify adding another. So in total there must be strictly larger than $(\ell-1)\frac{n}{k}$ agents. Similar to PRF, we can consider how cohesive $S$ is, and define ``near" to be within distance $D(S)$. 
\begin{defn}
    A set of representatives $X$ of size $k$ is said to satisfy {\em NORP} if for any set $S \subseteq X$ of size $\vert S \vert \geq \ell$ the following holds: $\vert v \in N:\exists x \in X \text{ s.t. } d(v,x) \leq D(S) \vert > (\ell-1)\frac{n}{k}$. 
\end{defn}
Informally, this means that any set of $\ell$ representatives needs to have at least one more agent nearby than enough agents to justify $\ell-1$ representatives.
This new fairness concept is orthogonal to all existing work, since it requires that sets are not {\em over-represented} instead of ensuring they are not {\em under-represented}.

\begin{table}[b]
\centering
\begin{tabular}{|l|l|l|l|l|}
\hline
Algorithm  & Sum Approximation         & PRF      & mJR      & NORP     \\ \hline
Algorithm \ref{alg:alg2b} & 2                     & 1        & 2        & $\infty$ \\ \hline
Algorithm \ref{alg:alg4} & 2                     & $\infty$ & $\infty$ & 1        \\ \hline
Algorithm \ref{alg:alg1} & 4                     & 1        & 1        & 1        \\ \hline
Algorithm \ref{alg:alg3} & 2                     & 2        & 2        & $\infty$ \\ \hline
Algorithm \ref{alg:alg5} & $2+\frac{2}{\alpha}$ & $\alpha$ & $\alpha$ & 1        \\ \hline
\end{tabular}
\caption{Summary of upper bounds given by the solutions $X$ given by each algorithm. The first column is the upper bound on the ratio $\frac{\cost(X)}{\cost(\OPT)}$. The other three columns give an upper bound on each of the fairness axioms we consider: 1 means that the property holds exactly, while anything larger means it holds approximately.}
\label{tab:uppers}
\end{table}

\subsection*{Our Contributions}
We begin with establishing lower bounds, showing that it is {\em not} possible to always form either a PRF or a NORP solution which is better than a $2$-approximation to the min-cost solution. We then show that this bound is tight, by providing an algorithm which forms a PRF and 2-approximate solution (Algorithm \ref{alg:alg2b}), and another which forms a NORP and 2-approximate solution (Algorithm \ref{alg:alg4}). All of our algorithms are polynomial-time, except for Algorithm \ref{alg:alg2b}, which is poly-time in many relevant special cases, and can be transformed into a PTAS. Thus, we show that it is possible to form low-cost PRF solutions, and low-cost NORP solutions, but what if we want all three at the same time?

Our main result (Section \ref{sec:main-res}) shows that we can efficiently compute a fully PRF and NORP solution, which is a $4$-approximation for the sum objective. This shows that all our fairness properties (as well as all the ones implied by PRF and mJR using \cite{kellerhals2023proportional}) are compatible with having a low-cost solution at the same time, although the approximation factor is no longer the best possible (4 instead of 2). Finally, we show how relaxing axioms of fairness can improve the cost of the solution, by providing algorithms that form a $2$-approximate $2$-PRF solution in guaranteed poly-time, and an algorithm that provides a  $2+\frac{2}{\alpha}$ approximate, $\alpha$-PRF and NORP solution all at once, allowing us to quantify the tradeoff between the cost of a solution and its level of fairness. We summarize the upper bounds given by our algorithms in Table \ref{tab:uppers}.

One take-away message from our work is that NORP is a desirable new fairness property, whose addition does not really hurt. It is still possible to form solutions which satisfy most of the existing fairness concepts, while simultaneously satisfying NORP and producing solutions with small overall cost (and thus high agent happiness). 


\section{Related Work}
A very large body of work exists on clustering and facility location in a metric space; this work is too much to survey here, see for example \cite{ackermann2010clustering,jain1999data, jain2010data, chang2016mathematical}
and references therein. In most of this clustering work, the goal is to choose facilities in order to minimize the distances from each agent to their {\em closest} facility (e.g., $k$-median, $k$-means, or $k$-center objectives). In our work, however, we focus on applications where the agents do not care only about their closest chosen facility, but about {\em all} chosen facilities, i.e., the sum objective. Note that for objectives like $k$-median, it is {\em not possible} to achieve fairness properties that we study (e.g., PRF or mJR) at the same time as a constant approximation to the desired objective. Consider an example on a line where $n - k + 1$ agents are located at zero and $k - 1$ agents are positioned at $1, 1 + \varepsilon, 1 + 2\varepsilon, \dots, 1 + (k - 2)\varepsilon$. To optimize \textit{k}-median, \textit{k}-means, or \textit{k}-center objectives, we would place $k - 1$ centers at $1, 1 + \varepsilon, 1 + 2\varepsilon, \dots, 1 + (k - 2)\varepsilon$ and one center at zero. However, this solution is inherently unfair: it under-represents the dense cluster on the left while over-representing the sparse cluster on the right. On the other hand, the sum objective permits both fairness and constant-factor approximations {\em simultaneously}.


As mentioned previously, the sum objective that we are interested in has been studied in the context of heterogeneous facility location problems. Such problems mostly have been considered in the mechanism design context, where the agents are strategic and can lie about their locations, see \cite{chan2021mechanism} for a survey. Because of this, the goal is usually to form algorithms which are {\em strategyproof}, i.e., the agents have no incentive to lie \cite{serafino2014heterogeneous, xu2021two, kanellopoulos2023discrete, zhao2024constrained, kanellopoulos2025truthful,deligkas2025agent}. 
Various related objectives are also used in existing work; in addition to our sum objective the distance to the farthest facility (e.g., \cite{lotfi2024truthful, deligkas2025agent}) or to the $q$'th closest facility (e.g., \cite{caragiannis2022metric,ebadian2024boosting}) have been studied. In contrast to most of this existing work, we do not assume that the agents are strategic, and instead focus on fairness properties in general metric spaces. Similar objectives have also been considered in the context of {\em metric distortion}, where the true locations of the agents are unknown and only their ordinal preferences are known, see e.g., \cite{goel2018relating,caragiannis2022metric}. Extending our results to other objectives is the subject of our future work.

Research on fairness in metric clustering problems has recently received a lot of interest in the fields of Machine Learning, AI, and Social Choice; see survey \cite{chhabra2021overview}. Many compelling fairness properties have been defined for this setting, including Proportional Fairness \cite{chen2019proportionally, micha2020proportionally}, Individual Fairness \cite{jung2020service}, and Transferable Core \cite{li2021approximate}. There has also been a significant amount of work (e.g., \cite{mahabadi2020individual, negahbani2021better,vakilian2022improved,bateni2024scalable}) forming approximately fair solutions that also approximate the best fair solution in the $k$-median or $k$-means objectives. Note that, as discussed above, it is {\em impossible} to form an approximately fair solution with respect to any of these fairness concepts that is also a good approximation to the optimum $k$-median solution; the work mentioned above forms an approximation with respect to the best {\em fair} solution. Unlike this existing work, we seek solutions which are a constant factor away from {\em any} solution for our objective, not just any {\em fair} solution.

While useful for many clustering applications, the fairness concepts mentioned in the previous paragraph have a major weakness: they require only a single representative to be close enough to each set of $\frac{n}{k}$ agents. For example, consider a line segment with 1 percent of agents at 0 and 99 percent at 1. If $k=100$ then the solution that places $99$ centers at 0 and 1 at 1 is considered fair by all the above definitions, despite greatly under-representing 99\% and over-representing 1\%. Because many settings, including the ones we discuss in the Introduction, desire some sort of proportional representation as well, several lines of work have generalized these fairness concepts to address this drawback. 
First, \cite{aziz2023proportionallyrepresentativeclustering} introduced PRF, as defined in the Introduction. They showed that PRF solutions always exist, and can be efficiently computed. We use the concept of PRF and show that it is compatible with the sum objective as well as NORP. Second, \cite{kellerhals2023proportional} defined metric versions of justified and proportional justified representation \cite{brill2023robust}, which they call mJR and mPJR respectively. They showed both always exist and can be computed via Expanding Approvals\cite{aziz2020expanding}. 
mPJR is a strengthening of mJR and implies PRF\footnote{A reader already familiar with mJR from \cite{kellerhals2023proportional} may wonder why we use mJR in our work, instead of mPJR, which extends mJR to sets of size $\ell\frac{n}{k}$. In fact, mPJR is incompatible with NORP, and we believe this is evidence that mPJR is too strong of a property: see Appendix \ref{app:mPJR}.}. They also showed various implications between mJR, mPJR, Proportional Fairness, Individual Fairness, Transferable Core, and $q$-core \cite{kellerhals2023proportional}. Importantly, they showed that mJR implies the best possible bounds for Proportional Fairness and  Individual Fairness, as well as bounds for Transferable Core. 
Since the solutions formed by our algorithms satisfy mJR, this immediately means that they approximately satisfy these other fairness properties also. Finally, \cite{kalayci2024proportional} considered the metric committee selection problem with the same motivation as ours: the agents care about more than just their closest representative. They define a new fairness property which they call {\em Proportionally Representative} (distinct from PRF from \cite{aziz2023proportionallyrepresentativeclustering}), which is a strengthening of Transferable Core. They show a variety of results and algorithms computing approximately fair solutions using their fairness concept. Unfortunately, exact Proportionally Representative solutions do not always exist, and resource augmentation (allowing fewer than the deserved amount $\ell$ representatives to be located near a set of size $\ell\frac{n}{k}$) is required to ensure that approximately fair solutions exist for this concept. While this existing work is very relevant to ours, we focus instead on forming fair solutions which are also compatible with low-cost solutions, and on satisfying NORP, a novel fairness concept which is orthogonal to all existing work, since it requires that sets are not {\em over-represented} instead of ensuring they are not {\em under-represented}.



\section{Model and Techniques} \label{sec:model}
We are given a set $N$ of $n$ agents in a pseudo-metric space with distance function $d$. 
The goal is to select a set $X$ of size $k$ from $N$. For ease of presentation, we will assume that $k$ divides $n$ in all our algorithms, so $\frac{n}{k}$ is an integer.
All of our results, however, also hold for the more general case of arbitrary $n$ and $k$. Rather than extending each of our results separately, we give a general reduction in Appendix \ref{app:kdividesn} which shows how to extend all of our algorithms and results for the case when $k$ does not divide $n$. 


We will now give some basic notation and definitions that we will use the rest of the paper.
We are interested in computing solutions $X$ that have low cost in the sum objective, which we define as
    $\cost_{(N,d)}(X) = \sum_{v \in N} \sum_{x \in X} d(v,x)$.
We will omit $N$ and $d$ when clear from context.
Likewise, the optimal solution $\opt$ is defined as 
    a set of size $k$ with minimum cost, i.e., the set $X$ which minimizes $\sum_{v \in N}\sum_{x \in X} d(v,x)$.
We will refer to any solution $X$ with $\cost(X)\leq \alpha\cdot\cost(\OPT)$ as an $\alpha$-approximation. 

Throughout the paper, we denote the diameter of a set $S$ by $D(S) = \max_{v,w \in S} d(v,w)$, and the radius by $R(S) = \min_{c \in N} \max_{v \in S} d(c,v)$, the size of the smallest ball centered at some agent in the entire collection of agents $N$ that contains the entire set.
 We will use the notation $B_S(x,\delta)=\{v \in S \mid d(v,x) \leq \delta\}$ for the set of points in $S$ that are in the closed ball of radius $\delta$ centered at $x$, and $B^\circ_S(x,\delta)=\{v \in S \mid d(v,x) < \delta\}$ for the open ball. When the set $S$ is clear from context or consists of the entire space, we will omit $S$.


We now review the axioms of fairness that we will consider in this paper, adjusted slightly for our setting in which any agent can be chosen as a representative. The main fairness concepts we use are PRF and NORP, already defined in the Introduction. Note that using the above notation, PRF means that there should be at least $\ell$ representatives in $\cup_{v \in S} B(v, D(S))$, and NORP means that there should be more than $(\ell-1)\frac{n}{k}$ agents in $\cup_{x \in S} B(x, D(S))$.
Likewise, we define mJR from \cite{kellerhals2023proportional}:

\begin{defn}
    An outcome $X \subseteq N$ with $\vert X \vert = k$ satisfies {\em metric Justified Representation} (mJR) if for all $S \subseteq N$ with $\vert S \vert = \frac{n}{k}$, we have that  $\cup_{v\in S}B(v, R(S))$ contains at least one element of $X$.
\end{defn}



We analogously define approximate fairness properties: $X$ satisfies $\alpha$-approximate PRF (which we will abbreviate to simply $\alpha$-PRF) if it satisfies PRF, but with $\alpha D(S)$ instead of $D(S)$. In other words, the needed representatives are allowed to be farther than members of $S$, but still not too far. Similarly, $\alpha$-mJR means satisfying mJR, but with distance requirement $\alpha R(S)$ instead of $R(S)$. 


\paragraph{Techniques}
Before presenting detailed results, we give an overview of the techniques used in our algorithms. 
All of our algorithms are similar to the Greedy Capture algorithm of \cite{chen2019proportionally}, the Truncated Greedy Capture algorithm of \cite{kalayci2024proportional}, and the Expanding Approvals algorithm of \cite{aziz2020expanding}. At the start of the algorithm, we begin with all agents in $N$ being {\em uncovered}. At a high level, all of our algorithms consist of selecting a set $S_i$ of agents of size $\frac{n}{k}$ at each iteration $i$, and then selecting a representative $x_i$ from this set. We will call agents in set $S_i$ as {\em covered} by set $S_i$. We will also sometimes refer to representative $x_i$ as {\em covering} all the agents in set $S_i$. Once an agent is covered in iteration $i$, we will never select them as a future member of a set $S_j$ for $j > i$. We continue in this manner until all agents are covered; since we cover exactly $\frac{n}{k}$ agents at each iteration, we end up with exactly $k$ representatives. 

Notice that simply running previously known versions of Greedy Capture or Truncated Greedy Capture (which choose $S_i$ to be the ball with smallest radius, and choose the center of this ball as $x_i$), can result in solutions which are not a good approximation to $\OPT$, and do not satisfy NORP. Because of this, we must be careful about how we choose sets $S_i$ and representatives $x_i$. Satisfying one or two of our desired properties is not too difficult, but as we add more constraints, choosing the right agent as representative becomes more complex, and requires a much more sophisticated analysis to show our bounds. Below we give a high-level sketch of how our algorithms behave; for full details see Sections \ref{sec:basic}, \ref{sec:main-res},  \ref{sec:approx}, and the Appendix. 

\begin{description}
    \item[Algorithm \ref{alg:alg2b}:] The set $S_i$ is the smallest-diameter uncovered set, and $x_i$ is the agent in $S_i$ closest to $\OPT$. When picking the smallest diameter sets we can select anything in $S_i$ as a representative and still have the solution be PRF; by selecting $x_i$ in $S_i$ that is closest to $\OPT$ we simultaneously get a 2-approximation to $\OPT$. 
    
    \item[Algorithm \ref{alg:alg4}:] The representative $x_i$ is the uncovered agent closest to $\OPT$. The set $S_i$ are the $\frac{n}{k}$ closest uncovered agents to $x_i$. By only selecting uncovered agents closest to $\OPT$ as representatives, we get a 2-approximation. Then by covering everything in a ball around it, we ensure that other representatives are far enough away that we achieve a NORP solution.
    
    \item[Algorithm \ref{alg:alg1}:] The set $S_i$ is the smallest-radius uncovered set, and $x_i$ is the agent in its center. Our algorithm is closely related to the many variants of Greedy Capture and Expanding Approvals algorithms from the literature (e.g., \cite{aziz2020expanding, aziz2021proportionally}), and is most similar to the Truncated Greedy Capture algorithm of \cite{kalayci2024proportional}. However, there are two crucial differences between our algorithm and all previous approaches in the literature. First, we stop growing balls for agents that are covered; this is crucial to avoid over-representation, as illustrated by Example \ref{ex:star}discussed in the Introduction.
    Unlike algorithms in existing work, this allows us to satisfy NORP and similar properties. Second, we move the last representative $x_k$ to the closest member of its ball from OPT, instead of leaving it in the center of the ball. This is necessary to make our solutions become a good approximation to OPT, as without this change our algorithm (as well as all algorithms in the literature, to the best of our knowledge) can result in an arbitrarily bad approximation ratio, see \ref{sec:main-res}. 
    Despite the difference in the choice of the last representative, we can still show that this algorithm achieves PRF and mJR, while also satisfying NORP and guaranteeing a 4-approximation to OPT.

    
    \item[Algorithm \ref{alg:alg3}:] The sets $S_i$ are the same as in Algorithm \ref{alg:alg1} but we move all representatives to the agent in $S_i$ closest to $\OPT$. This can violate NORP, and is only 2-PRF, but gives a 2-mJR solution and a 2-approximate solution in polynomial time. 
    
    \item[Algorithm \ref{alg:alg5}:] Grows the sets $S_i$ in the same way as Algorithm \ref{alg:alg1}, except for the uncovered agent that is closest to $\OPT$ that has its radius  grow at a rate $\alpha$ times faster than all other agents. $x_i$ is selected the same as in Algorithm \ref{alg:alg1}. This gives a $2+\frac{2}{\alpha}$ approximate and $\alpha$-PRF solution, while also satisying NORP.
\end{description}

\section{Basic Compatibility of Low-Cost and Fairness} \label{sec:basic}
We begin by establishing basic impossibility results: satisfying either PRF or NORP (even approximately) is not compatible with forming a better than 2-approximation for our cost objective.\footnote{In fact, it is co-NP-hard to determine if a PRF solution exists which is better than a 2-approximation to OPT, see Appendix \ref{sec:hardness}.}
\begin{prop}\label{prop:2-inapprox}
     For any $\epsilon > 0$, there exists an instance where all solutions with cost at most $(2-\epsilon)\cost(\OPT)$ do not satisfy $\alpha$-PRF (for any finite $\alpha \geq 1$) or NORP.
\end{prop}

\begin{proof}
    See the example in Figure \ref{fig:2approxnotpossible}. In this example assume that $k = \sqrt{n}$ and $n$ is a perfect square. We have $k$ agents located at each leaf of a star graph, and $k$ at the center; the distances from the center to each leaf are all equal to 1, and there are $k-1$ leaves. Consider any set of agents of size $\frac{n}{k}=k$ located at the same point. This is a set of diameter zero, so for a solution to be $\alpha$-PRF (for any finite $\alpha$), it must place at least one representative at that point. Thus, we see that for a solution to be $\alpha$-PRF we must place a representative at each arm of the star graph, and one at the center. Likewise, NORP would allow at most 1 representative in the center, since otherwise the 2 representatives in the center of the graph would constitute a set with diameter zero, but with only $\frac{n}{k}$ agents distance zero away from them. We have $k-1$ arms. Thus the cost of any $\alpha$-PRF or NORP solution is $k(k-1)(1+2(k-2))+k(k-1)>2k(k-1)(k-2)$. 
    On the other hand, $\OPT$ places all representatives in the center, thus the cost of $\OPT$ is $k^2(k-1).$ Thus the approximation ratio between the cost of any $\alpha$-PRF or NORP solution, and the cost of $\OPT$, approaches $2$ as $n$ and $k$  approach $\infty$.
\end{proof}

\begin{figure}
    \centering
    \begin{tikzpicture}

        \node (center) at (0,0) {k agents};

        \foreach \i in {1,2,...,7} 
            \node (node_\i) at ({360/7*\i}:3) {k agents};

        \foreach \i in {1,2,...,7} 
            \draw[decorate, decoration={segment length=5mm, amplitude=1mm}] 
            (center) -- (node_\i) node[pos=0.5, fill=white] {1};
    \end{tikzpicture}
    \caption{Example illustrating the incompatibility of PRF, NORP, and $(2-\epsilon)$-approximate solutions.}
    \label{fig:2approxnotpossible}
\end{figure}
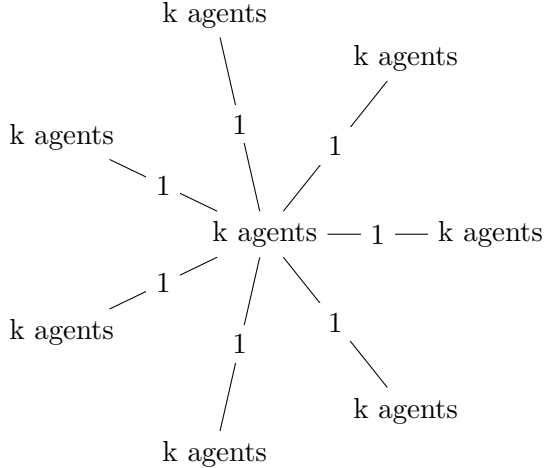

Our main result in this paper is given in Section \ref{sec:main-res}, where we give an algorithm that provides a simultaneously PRF, NORP, and mJR solution that is at worst a 4-approximation to $\OPT$. In the rest of this section, however, we first consider PRF and NORP individually, and show that the above lower bound is tight by proving that there always exists a 2-approximate solution satisfying PRF, and the same is true for NORP. The algorithms to compute these solutions, and the proofs of these results, are comparatively simple, although we will use ideas from both algorithms in Section \ref{sec:main-res}.

First, we will introduce an algorithm that gives us a solution satisfying PRF that simultaneously is also at most $2$ times the cost of the optimal solution. The algorithm proceeds as follows: we start with all agents ``uncovered" then over the course of the algorithm we consider agents covered when a representative is placed nearby. The algorithm works by repeatedly finding the smallest diameter subset of the uncovered agents with size $\frac{n}{k}$. Then we select the agent in the set that minimizes its average distance to $\OPT$. By finding these aforementioned small sets, we ensure that PRF is maintained, and we also ensure that an efficient and low-cost solution is found by selecting representatives in each set that are relatively close to the optimal solution, $\OPT$. Unfortunately, in general, implementing our algorithm is NP-Hard for general metric spaces, via a trivial reduction from $Clique$. (We will later give a polynomial time approximately fair algorithm, as well as a perfectly fair 4-approximate solution). Fortunately our algorithm is polynomial time when the set of points lies in $\mathbb{R}^2$ with the $L_2$ metric \citep{aggarwal1989fining} and is fixed parameter tractable for the $L_\infty$ metric \citep{eppstein1994iterated}. In three dimensions there is also a PTAS that for any $\epsilon > 0$ finds a set that has diameter at most $(1+\epsilon)$ away from the diameter of the smallest set in time $O(n \log n + 2^{O(1/(\epsilon^{1.5}\log \epsilon))})$ \citep{jiang2007finding}. To the best of our knowledge, it is unknown whether the problem of selecting a subset of points of $k$ elements in $R^n$ with the norm $L_2$ that minimizes the diameter is NP-Hard. Note that if we use the PTAS given by \citep{jiang2007finding}, we get a solution that is still a two-approximation for our objective, but is an approximate $(1+\epsilon)$-PRF solution.
\begin{algorithm}
\caption{Forms 2-approximate, 1-PRF solution}\label{alg:alg2}
\begin{algorithmic}[1]
\State $Z \gets N$ \Comment{The set of agents that are uncovered so far}
\While{$Z \neq \emptyset$}
    \State $S_i \gets \arg \min_{\substack{S \subseteq Z \\ \vert S \vert = \frac{n}{k}}} \max_{j,j' \in S} d(j,j')$ \Comment{Greedily find the agents that we cover}
    \State $Z \gets Z \setminus S_i$ \Comment{Remove the agents that we covered}
    \State $\delta_i \gets \max_{j,j' \in S_i} d(j,j')$ \Comment{The diameter of the set that we found }
    \State $s_i \gets \arg \min_{v \in S_i} \sum_{o \in \OPT} d(v,o)$ \Comment{Pick the agent that is closest to $\OPT$}
\EndWhile
\State $X \gets \{s_i \mid i \in 1,...,k\}$
\end{algorithmic}
\end{algorithm}

In the above algorithm, as in all our algorithms, if there is a tie for selecting a set $S_i$, or a representative $s_i$, then one can be chosen arbitrarily. The algorithm will give us all the desired properties with any tie-breaking rule. We will now introduce some terminology that we will use frequently in the following proofs. Similar concepts will be used in the proofs for other algorithms later on. We will call the set $Z$ consisting of a subset of agents the uncovered agents. Notice that in each iteration $Z$ has the set $S_i$ removed from it, and $S_i \subseteq Z$ before the set is removed from $Z$ on line 4. Therefore, in each iteration we remove $\frac{n}{k}$ agents, effectively covering these agents. Likewise, we will also refer to representatives that cover agents. We say that a representative $s_i$ covers an agent $v$ if $v$ was covered on the same iteration as in which $s_i$ was selected, or equivalently $v \in S_i$. The next fact to note is that the values of $\delta_i$ are non-decreasing. This fact leads to an equivalent formulation of Algorithm \ref{alg:alg2} as a continuous process given by Algorithm \ref{alg:alg2b}.

\begin{algorithm}
\caption{Forms 2-approximate, 1-PRF solution continuously}\label{alg:alg2b}
\begin{algorithmic}[1]
\State $Z \gets N$ \Comment{The set of agents that are uncovered so far}
\State $\delta \gets 0$
\While{$Z \not = \emptyset$}
\While{$\{S \mid (S \subseteq Z) \land (\vert S\vert = \frac{n}{k}) \land (D(S) = \delta)\} = \emptyset$}
    \State $\text{smoothly increase }\delta$
\EndWhile
    \State $S_i \in \{S \mid (S \subseteq Z) \land (\vert S\vert = \frac{n}{k}) \land (D(S) = \delta)\}$
    \State $Z \gets Z \setminus S_i$ \Comment{Remove the agents that we covered}
    \State $\delta_i \gets \delta$
    \State $s_i \gets \arg \min_{v \in S_i} \sum_{o \in \OPT} d(v,o)$ \Comment{Pick the agent that is closest to $\OPT$}
\EndWhile
\State $X \gets \{s_i \mid i \in 1,...,k\}$
\end{algorithmic}
\end{algorithm}

The continuous formulation of the above algorithm is presented in Algorithm \ref{alg:alg2b}. The algorithm starts with an initial diameter $\delta = 0$, and continuously increases the diameter until a set of size $\frac{n}{k}$ is found that has a diameter of $\delta$. Notice that the next set we find is always precisely one of the sets in $\arg \min_{\substack{S \subseteq Z \\ \vert S \vert = \frac{n}{k}}} \max_{j,j' \in S} d(j,j')$. The continuous version allows us to introduce the notion of time. We can treat $\delta$ as the time. In our proofs and discussion when we refer to time $t$ we refer to the point in the running time of the continuous formulation of Algorithm \ref{alg:alg2} when $\delta=t$. For the rest of the paper we will present algorithms as continuous processes, but in the same way that we have done here, we can convert them to discrete algorithms that require $k$ iterations. 
Now we can move on and show the fairness and efficiency guarantees of Algorithm \ref{alg:alg2}.

\begin{lem} \label{lem:alg2prf}
    The solution $X$ given by Algorithm \ref{alg:alg2} satisfies PRF.
\end{lem}

\begin{proof}
    
    Let $S \subseteq N$ be a set of agents of size $\ell \frac{n}{k}$ with diameter $D(S)$.
    Consider two cases: the first where $D(S) < \delta_k$, and the case where $D(S) \geq \delta_k$.
    
    First, we consider the case where $D(S) < \delta_k$. Then after time $D(S)$, it must be
    that fewer than $\frac{n}{k}$ agents are left uncovered in $S$. This is because otherwise we would have added an additional representative to cover some of the set $S$, since any subset of $S$ has diameter at most $D(S)$. Notice that each chosen representative can cover at most $\frac{n}{k}$ agents in $S$, and therefore there are at least $\ell$ representatives that cover something in $S$ after time $D(S)$. For any $s_i$ that covers some agent $v\in S$ by time $D(S)$, it must be that $d(v,s_i)\leq \delta_i\leq D(S)$, since $s_i$ was chosen at time $\delta_i$ before time $D(S)$. Thus, all of these $\ell$ representatives are at most $D(S)$ away from something in $S$, and so PRF is satisfied for $S$.

    Now suppose $D(S) \geq \delta_k$. At the end of the algorithm, all agents are covered. Since each representative can only cover $\frac{n}{k}$ agents, and each representative is at most distance $\delta_k$ from the agents they cover, it must be that there are at least $\ell$ representatives within distance $D(S)$ from something in $S$. 
    \end{proof}

Now we will show that $\cost(X) \leq 2 \cdot \cost(\OPT)$. The proof is quite general and simple. So, first, we will state a general proposition from which the cost bounds follow directly.
\begin{prop}
    Let $X$ be a set of representatives, and $P = \{P_1, P_2,...,P_k\}$ be a partition of $N$ into $k$ sets of equal size. Then if $X = \{\arg \min_{v \in P_i} \sum_{o \in \OPT}d(v, o) \mid P_i \in P\}$ (with ties broken arbitrarily), then $\cost(X) \leq 2 \cdot \cost(\OPT)$. \label{prop:2approx}
\end{prop}

\begin{proof}
    Let $x_i \in X$ be the representative in $X$ such that $x_i \in \arg \min_{v \in P_i} \sum_{o \in \OPT}d(v, o)$. Then we have that the cost of $X$ is given by the following:
    \begin{align}
        \sum_{x_i \in X} \sum_{v \in N} d(v, x_i) &\leq \frac{1}{k} \sum_{x_i \in X} \sum_{v \in N}   \sum_{o \in \OPT} \left [d(v, o) + d(x_i,o)\right]\\
        &= \cost(\OPT) + \frac{n}{k} \sum_{x_i \in X} \sum_{o \in \OPT} d(x_i,o)\\
        &= \cost(\OPT) + \frac{n}{k} \sum_{x_i \in X} \sum_{o \in \OPT} \frac{k}{n}\sum_{v \in P_i}d(x_i,o)\\
        &= \cost(\OPT) + \sum_{x_i \in X} \sum_{o \in \OPT} \sum_{v \in P_i}d(x_i,o)
        \end{align}
        where the first line comes from applying the triangle inequality $d(v,x_i) \leq d(v,o) + d(o,x_i)$ for all $o \in \OPT$ simultaneously (and averaging the $k$ inequalities). We recall that the bound $\sum_{o \in \OPT}d(v, o) \geq \sum_{o \in \OPT} d(x_i, o)$ holds for all $v \in P_i$, since we chose an element of $P_i$ which minimizes $\sum_{o \in \OPT}d(v, o)$ as $x_i$. Thus we obtain the following bounds:
        \begin{align}
        \cost(\OPT) + \sum_{x_i \in X} \sum_{o \in \OPT} \sum_{v \in P_i}d(x_i,o) 
        &\leq \cost(\OPT) + \sum_{x_i \in X} \sum_{v \in P_i}\sum_{o \in \OPT} d(v,o)\\
        &=\cost(\OPT) + \sum_{v \in N} \sum_{o \in \OPT} d(v,o)\\
        &= 2\cost(\OPT).
    \end{align}
    Thus, $\cost(X) \leq 2 \cdot \cost(\OPT)$.
\end{proof}

\begin{cor}
    The solution $X$ given by Algorithm \ref{alg:alg2} has cost $\cost(X) \leq 2 \cdot \cost(\OPT)$.
\end{cor}
\begin{proof}
    Notice that the sets $S_i$ form a partition and $s_i \in \arg\min_{v \in S_i} \sum_{o \in \OPT} d(v,o)$.
\end{proof}
Next, we will show a useful reduction lemma that allows us to not have to repeatedly prove approximate mJR bounds. We show that PRF implies a $2$ approximation to mJR.
\begin{lem} \label{lem:prftomJR}
    If $X$ is $\alpha$-PRF then $X$ is 2$\alpha$-mJR.
\end{lem}
\begin{proof}
Let $X$ be a PRF solution. Then let $S \subseteq N$ be a subset of size $\frac{n}{k}$. Recall that the radius of $S$ is defined as the smallest $\delta$ such that there is an agent $v \in N$ with $B(v, \delta) \supseteq S$. We then need to show that there is a representative in $\cup_{v \in S} B(v, \alpha R(S))$ for $X$ to satisfy $\alpha$-mJR.

Let $v$ be the agent in $N$ such that $B(v, R(S)) \supseteq S$, and $v_1$ and $v_2$ be the agents in $S$ such that $d(v_1, v_2) = \max_{v,v' \in S} d(v,v')$. Note that $D(S) = d(v_1, v_2) \leq d(v_1, v) + d(v,v_2) \leq 2 R(S)$. Thus, we have $D(S) \leq 2 R(S)$, and by $\alpha$-PRF there is a representative in $\cup_{v \in S} B(v,\alpha D(S)) \supseteq \cup_{v \in S} B(v, 2\alpha R(S))$.
\end{proof}

\begin{cor}
    The solution returned by Algorithm \ref{alg:alg2} is 2-mJR.
\end{cor}
\begin{proof}
    Follows from Lemma \ref{lem:alg2prf} and Lemma \ref{lem:prftomJR}. 
\end{proof}
We have now seen an algorithm which returns the best possible approximation for OPT (assuming we want some sort of fairness), while forming a PRF solution. Unfortunately, this algorithm may return a solution that is not NORP. To see this, consider the example in Figure \ref{fig:notnorpalg2}. Algorithm \ref{alg:alg2} would select each arm and the closest central agent at each iteration, then place the representative in the central region. Thus, there are $6$ representatives with a diameter of $\epsilon$ and only $6$ agents within $\epsilon$ of a representative in the solution.

\begin{figure}[b]
    \centering    
    \begin{tikzpicture}
\usetikzlibrary{math} 

\foreach \i in {1,2,...,6} {
    \node[circle,fill,inner sep=1.5pt] (node_\i) at ({360/12+360/6*\i}:.5) {};
}

\foreach \i in {1,2,...,6} {
    \node (node2_\i) at ({360/12+360/6*\i}:3) {$\frac{n}{k}-1$ agents};
    }

\foreach \i in {1,2,...,6} { 
    \foreach \j in {1,...,6} {
        \ifthenelse{\i < \j}{\draw[decorate, decoration={segment length=5mm, amplitude=1mm}] 
        (node_\i) -- (node_\j);
        }}{}
}

\foreach \i in {1,2,...,6} {
    \draw[decorate, decoration={segment length=5mm, amplitude=1mm}] 
    (node2_\i) -- (node_\i) node[pos=0.5, fill=white] {1};
}
\end{tikzpicture}
    \caption{Example where Algorithm \ref{alg:alg2} gives a solution that is not NORP. $k=6$ and there are $6$ central agents that are all a distance $\epsilon$ from each other; each central agent has a spoke with $\frac{n}{k}-1$ agents a distance $1$ away, all of which are also distance $\epsilon$ from each other. Previous fairness axioms would allow representatives to be placed in the inner clique. This is unfair since the $k$ agents in the middle have a large number of representatives within a distance of $\epsilon$, yet only $6$ agents. A fairer solution (that satisfies all usual fairness properties) would be to place the representatives on the outermost $5$ points with $\frac{n}{k}-1$ agents and one representative in the middle.}
    \label{fig:notnorpalg2}
\end{figure}
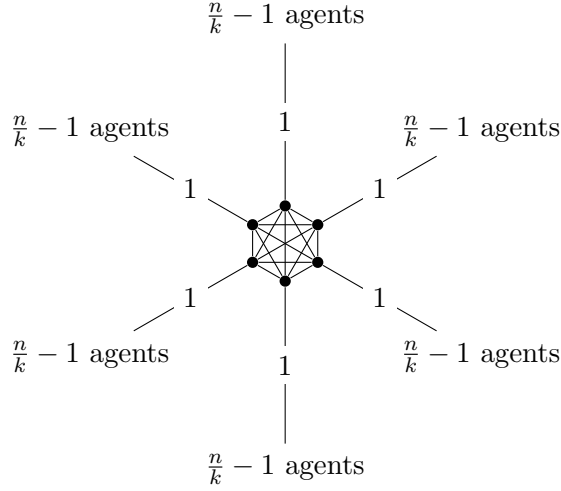


Now that we have shown that PRF is compatible with a $2$-approximate solution, we ask the same question about NORP: is it always possible to form a $2$-approximate solution which is NORP? Intuitively, satisfying NORP can be easy: simply place representatives far away from each other. However, requiring a solution which is both low-cost and NORP is trickier, since a low cost solution should have most of its representatives be centrally located. Thus, we need the representatives to be centrally located while also avoiding overrepresentation.

The next algorithm works by repeatedly growing the ball centered on the uncovered agent $v$ in $\OPT$ that minimizes $\sum_{o \in \OPT} d(v,o)$. Then, when the ball contains $\frac{n}{k}$ uncovered agents, we add the agent at the center of the ball to the set of representatives and cover the $\frac{n}{k}$ agents closest to $v$. Thus, the algorithm runs for $k$ iterations and places $k$ representatives, each of which covers $\frac{n}{k}$ agents. We therefore have a 2-approximate solution and a NORP solution simultaneously in polynomial time. By removing the closest $\frac{n}{k}$ agents from never being a representative in the future we guarantee that the solution will have a so-called ``nesting" property, meaning for any two representatives $x$ and $x'$, if $x'$ is added after $x$, there will be $\frac{n}{k}$ agents with distance to $x$ less than $d(x, x')$. 
\begin{algorithm}
\caption{$(2, \infty, 1)$ algorithm}\label{alg:alg4}
\begin{algorithmic}[1]
\State $Z \gets N$ \Comment{Uncovered Agents}
\State $\delta \gets 0$
\State $i \gets 1$
\While{$Z \neq \emptyset$}
\State $s_i \gets \arg \min_{v \in Z} \sum_{o \in \OPT}d(v,o)$ \Comment{Agent closest to $\OPT$, whose ball we grow}
\While{$\vert B(s_i, \delta) \cap Z\vert < \frac{n}{k}$}
\State $\text{smoothly increase } \delta$
\EndWhile
\State $S_i \gets \arg \min_{\substack{T \subseteq B(s_i, \delta) \cap Z \\ \vert T \vert = \frac{n}{k}}} \sum_{v \in T} d(v, s_i)$
\State $Z \gets Z \setminus S_i$
\State $\delta_i \gets \delta$
\State $i \gets i + 1$
\EndWhile

\State $X \gets \{ s_i \mid i \in 1,...,k \}$
\end{algorithmic}
\end{algorithm}
Like previously, we define the set $Z$ as the set of uncovered agents and $s_i = x_i$ to cover $v$ in the same iteration when $s_i$ was selected, $v$ was covered. Also, like previously, we refer to $\delta$ as the time in the continuous version of the algorithm.

In addition, we observe that the returned solution may not be $\alpha$-PRF for any $\alpha>0$. To see this, consider the example in Figure \ref{fig:notprf}.

\begin{figure}[b]
    \centering
    \begin{tikzpicture}[
        group/.style={rectangle, draw, minimum width=2cm, minimum height=1cm, font=\large},
        dot/.style={fill=black, circle, inner sep=1.5pt},
        distancelabel/.style={font=\small}
    ]

    \node[group] (left) {$K_{65}$};
    \node[below=0.3cm of left] {Internal distances: $2-\epsilon$};

    \node[group, right=6cm of left] (right) {$K_{34}$};
    \node[below=0.3cm of right] {Internal distances: $\epsilon$};

    \coordinate (midpoint) at ($(left.east)!0.5!(right.west)$);
    \coordinate (leftMid) at ($(left.east)!0.25!(right.west)$);
    \coordinate (rightMid) at ($(left.east)!0.75!(right.west)$);

    \draw[-] (left.east) -- (right.west);

    \node[dot] at (midpoint) {};

    \node[distancelabel, above=0.25cm of leftMid] {$1$};
    \node[distancelabel, above=0.25cm of rightMid] {$1+\epsilon$};

    \draw[-] (left.north) to[out=30, in=150] node[above, yshift=0.05cm] {\small $2-2\epsilon$} (right.north);

    \node[below=0.1cm of midpoint] {\footnotesize 3 agents};

    \end{tikzpicture}
    \caption{Example where Algorithm \ref{alg:alg4} does not satisfy PRF. There are 102 agents in total, broken into three groups. Three agents in the center on top of each other, 65 on the left and 34 on the right. We use $K_r$ to denote a complete graph with $r$ nodes. The 65 on the left are $2-2\epsilon$ from each other and the 34 on the right are $\epsilon$ from each other. All agents in $K_{65}$ and $K_{34}$ are $2-2\epsilon$ from each other. The three agents in the middle are $1$ and $1+\epsilon$ away from the two groups. The three middle agents are the optimal solution for $k=3$ and $\frac{n}{k} = 34$. The algorithm instead picks one central representative and two on the left. Then the group on the right does not have a representative within distance $\epsilon$, which is the diameter of that set.}
    \label{fig:notprf}
\end{figure}
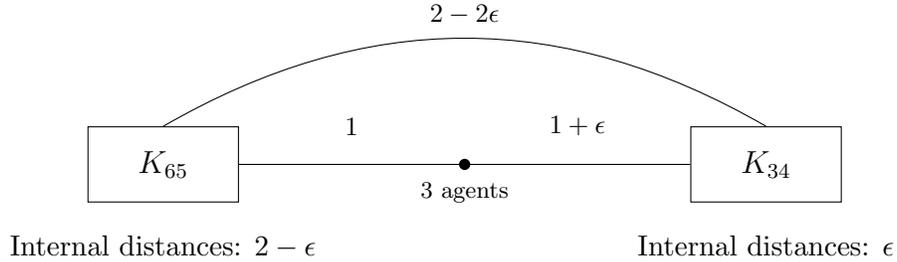

\begin{thm}
    Solution $X$ formed by Algorithm \ref{alg:alg4} is NORP and has cost $\cost(X) \leq 2 \cdot\cost(\OPT)$.
\end{thm}
We now present our first proof that a solution satisfies NORP. The high-level idea is that each representative has a ball around it when it is selected. Then by selecting representatives that are not inside each other's balls we ``nest" them so that the representatives are some minimum distance from each other, ensuring that $D(S)$ is sufficiently large for any subset $S$ of $X$.
\begin{lem}
    Solution $X$ formed by Algorithm \ref{alg:alg4} satisfies NORP.
\end{lem}
\begin{proof}
Let $S \subseteq X$ be a set of representatives of size $\ell>1$ (or else NORP is trivially satisfied). Note that due to the algorithm, each representative $s_i$ that is chosen has a corresponding value $\delta_i$; the set $S_i$ of agents which are covered by $s_i$ is guaranteed to lie in $B(s_i,\delta_i)$.

Let $x$ and $y$ be arbitrary representatives in the set $S$. Let $\delta_x$ and $\delta_y$ be the values of $\delta$ when $x$ and $y$ were selected as representatives in the algorithm. To prove our result, we first prove that $D(S) \geq \min(\delta_x,\delta_y)$. 

Suppose that $x$ was added after $y$. Then we know that $x \not \in B^\circ(y, \delta_y)$, or otherwise we would have covered $x$ when selecting $y$. Thus, we have $d(x, y) \geq \delta_y$. Since $y$ may have been added after $x$ instead, this means that we always have either $d(x, y) \geq \delta_y$ or $d(x, y) \geq \delta_x$. Together, this means that $D(S) \geq \min(\delta_x,\delta_y)$.

To prove the lemma, we need to show that there are enough agents near $S$ to justify having $\ell$ representatives. Consider the two representatives in $S$ with the first and second largest value of $\delta_i$; we will call these representatives $x_l$ and $x_s$ (respectively), and their corresponding $\delta_i$ values will be denoted by $\delta_l$ and $\delta_s$. By the above argument, we know that $D(S)\geq \min(\delta_s,\delta_l)= \delta_s$. Thus, if we can prove that $\cup_{s_i \in S} B(s_i, \delta_s)$ contains a lot of agents, then we are done, since $\cup_{s_i \in S} B(s_i, \delta_s)\subseteq \cup_{s_i \in S} B(s_i, D(S))$. 

Let $S'=S\setminus\{x_l\}$. We know that $\delta_i \leq \delta_s$ for all $s_i\in S'$, since $\delta_s$ is the second largest $\delta$ value in set $S$. Recall that each ball centered on representative $s_i \in S'$ covered $\frac{n}{k}$ agents, all of which were at most $\delta_i$ away from $s_i$. In addition, no agent was removed by two different sets, so for all $s_i \in S'$ the sets they cover are disjoint. So, there are at least $(\ell-1)\frac{n}{k}$ agents in $\cup_{s_i \in S'} B(s_i, \delta_s)$, and thus in $\cup_{s_i \in S'} B(s_i, D(S))$. Now consider the ball of radius $D(S)$ centered on the last representative, $x_l$. There is at least one additional agent that is covered by representative $x_l$ located at $x_l$, since we chose $x_l$ as a representative and we only choose uncovered agents as representatives. Thus, we have at least $(\ell-1)\frac{n}{k}+1$ agents in $\cup_{x \in S} B(x, D(S))$ for every $S$, so $X$ satisfies NORP. 
\end{proof}

The following lemma follows directly from an application of Proposition \ref{prop:2approx}. However, the fact that NORP is compatible (with such a low approximation factor) with forming a low-cost solution is somewhat surprising. The specific cost function we consider here ends up with $\OPT$ being a dense low-diameter set, while NORP requires that the solution is the opposite: sparse and spread out. Despite the seeming incompatibility, we arrive at a simple approximation bound, that is also the best possible according to Proposition \ref{prop:2-inapprox}.
\begin{lem}
     Solution $X$ formed by Algorithm \ref{alg:alg4} has cost $\cost(X) \leq 2 \cdot\cost(\OPT)$.
\end{lem}
\begin{proof}
    The proof follows  immediately from Proposition \ref{prop:2approx}.
\end{proof}

\section{Main Result: Low-cost, Fair, and NORP at the same time} \label{sec:main-res}
Now that we have an algorithm that provides a low-cost fair solution, and another that gives a low-cost NORP solution, we will attempt to combine the ideas of both algorithms and introduce our next algorithm, which will do both \textit{simultaneously}. Algorithm \ref{alg:alg1} gives a solution that satisfies all our desired fairness concepts: PRF, NORP, and mJR. At the same time, the solution it forms has cost that is at most four times that of the optimal solution. Our algorithm starts by smoothly growing balls around all agents in $N$; if multiple agents are located at the same point, we grow each ball independently. Eventually, when a ball has at least $\frac{n}{k}$ uncovered agents, we add the agent whose ball contained $\frac{n}{k}$ total uncovered agents to our set. We then greedily cover agents from the closest to the center outward (including the agent that grew the ball), so that the only agents that are not covered must be at least $\delta_i$ (the radius of the ball when $x_i$ was added) away from the center $x_i$. We then stop growing the balls around all the agents that are covered. We continue in this manner: each time there is a ball with agents uncovered $\frac{n}{k}$, we then add the representative at the center of this ball and remove the closest $\frac{n}{k}$ from it from ever becoming a center. Recall that $k$ divides $n$ so when we place $k$ representatives we would have removed $k \cdot \frac{n}{k} = n$ agents; we then stop growing the balls when we cover all agents. Lastly, we move the last center $x_k$ to the agent in the set that $x_k$ covered that minimizes its average distance to the optimal solution. Note that we present our algorithm as a continuous process, although it is easy to discretize our algorithm and transform it into one that is polynomial in time, similarly to the technique in  \citet{chen2019proportionally}.

\begin{algorithm}
\caption{(4,1,1) algorithm}\label{alg:alg1}
\begin{algorithmic}[1]
\State $Z \gets N$
\State $\delta \gets 0$
\State $i \gets 1$
\While{$Z \neq \emptyset$}
\While{$\exists x \in Z: \vert B(x, \delta) \cap Z\vert \geq \frac{n}{k}$}
\State $S_i \gets \arg \min_{\substack{T \subseteq B(x, \delta) \cap Z \\ \vert T \vert = \frac{n}{k}}} \sum_{v \in T} d(v, x)$
\State $Z \gets Z \setminus S_i$
\State $s_i \gets x$
\State $\delta_i \gets \delta$
\State $i \gets i + 1$
\EndWhile
\State $\text{smoothly increase } \delta$
\EndWhile
\State $X \gets \{s_i \vert i < k\} \cup \{ \arg \min_{v \in S_k} d(v, \OPT)\}$
\end{algorithmic}

\end{algorithm}

Our algorithm is closely related to the many variants of Greedy Capture and Expanding Approvals algorithms from the literature (e.g., \citet{aziz2020expanding, aziz2021proportionally}), which are most similar to the Truncated Greedy Capture algorithm of \citet{kalayci2024proportional}. However, there are two crucial differences between our algorithm and all previous approaches in the literature. First, we stop growing balls for agents that are covered; this is crucial to avoid over-representation as depicted in Example \ref{ex:star}. Unlike algorithms in existing work, this allows us to satisfy NORP and similar properties. Second, we move the last representative to the closest member of its ball from OPT, instead of leaving it in the center of the ball. This is necessary to make our solutions become a good approximation to OPT, as without this change our algorithm (as well as all algorithms in the literature, to the best of our knowledge) can result in an arbitrarily bad approximation ratio. For example, consider the simple example in Figure \ref{fig:linebadapprox}: running variants of Greedy Capture without moving the last center results in a solution that is arbitrarily bad compared to the optimal solution.

\begin{figure}[b]
    \centering
     \begin{tikzpicture}
        \node (left) at (0,0) {$n-2$ agents};

        \node (mid) at (6,0) {$1$ agent};
        \node (right) at (10,0) {$1$ agent};

        \draw[decorate, decoration={segment length=5mm, amplitude=1mm}]
        (left) -- (mid) node[pos=0.5, fill=white] {1};

        \draw[decorate, decoration={segment length=5mm, amplitude=1mm}]
        (right) -- (mid) node[pos=0.5, fill=white] {$\epsilon$};

    \end{tikzpicture}
    \caption{In the above example we have $n-2$ agents on the left hand side each distance $0$ from each other and a single agent each in the middle and on the right. Standard Truncated Greedy Capture (along with Greedy Capture and Expanding Approvals if we place extra representatives on the left) would first place $k-1$ representatives on the left, then the last representative in the middle. The optimal solution, however, places all representatives on the left. The cost of the solution returned by the greedy capture algorithm is at least $(k-1)(2-\epsilon) + (n-2)(1-\epsilon)$. The cost of the optimal solution is $k(2-\epsilon)$. As $n \to \infty$ the approximation factor becomes arbitrarily large.}
    \label{fig:linebadapprox}
\end{figure}

\begin{thm} \label{thm:31infty}
    Let $X$ be the solution returned by Algorithm \ref{alg:alg1}. Then $X$ satisfies PRF, mJR, and NORP, with $cost(X)\leq 4\cdot cost(OPT)$.
\end{thm}

Note that the sets $S_i$ form a partition of the agents $N$, since at each iteration we remove $S_i$ from $Z$, and each $S_i$ of exactly the size $\frac{n}{k}$. In addition, each $S_i$ has exactly one representative $x_i$ of $X$ as a member: for $i<k$ this is just $x_i=s_i$, and for $i=k$ this is the agent $x_k\in S_k$ that minimizes the distance to OPT. As described above, we will refer to an agent $v$ as {\em covered} by the set $S_i$ (and representative $x_i$) if this agent $v\in S_i$. We will also refer to the solution at {\em time $t$} to mean the sets $S_i$ and centers $s_i$ chosen so far by Algorithm \ref{alg:alg1} when $\delta=t$, since the value of $\delta$ is continuously increasing. We will refer to the set $Z$ at time $\delta$ as {\em uncovered} agents.

    


The proof that Algorithm \ref{alg:alg1} satisfies PRF follows the same logic as the PRF proof for Algorithm \ref{alg:alg2b}. 
 \begin{lem}
    The solution $X$ of Algorithm \ref{alg:alg1} satisfies PRF. \label{lem:alg1prf}
\end{lem}
\begin{proof}
    Let $S \subseteq N$ be a set of agents of size $\ell \frac{n}{k}$ with diameter $D(S)$. Consider two cases: the first where $D(S) < \delta_k$, and the case where $D(S) \geq \delta_k$.

    First, we consider the case where $D(S) < \delta_k$. Then after time $D(S)$, it must be
    that fewer than $\frac{n}{k}$ agents are left uncovered in $S$. This is because otherwise we would have added an additional representative to cover some of the set $S$, since any subset of $S$ has diameter at most $D(S)$, and thus a ball centered at one of the uncovered agents with radius $D(S)$ would contain all uncovered agents of $S$. Notice that each chosen representative can cover at most $\frac{n}{k}$ agents in $S$, and therefore there are at least $\ell$ representatives that cover something in $S$ after time $D(S)$. Since $D(S) < \delta_k$, all the representatives chosen by time $D(S)$ are centers of the chosen balls, since we only move the last representative in $X$. Thus, for any $s_i\in X$ that covers some agent $v\in S$ by time $D(S)$, it must be that $d(v,s_i)\leq \delta_i\leq D(S)$, since $s_i$ was chosen at time $\delta_i$ before time $D(S)$. Thus, all of these $\ell$ representatives are at most $D(S)$ away from something in $S$, and so PRF is satisfied for $S$.

    Now suppose $D(S) \geq \delta_k$. At the end of the algorithm, all agents are covered.
    Since each representative can only cover $\frac{n}{k}$ agents, it must be that there are at least $\ell$ representatives covering something in $S$.
    Note that with the exception of the last representative, all representatives are at most $\delta_k$ away from all agents in their covered set, since representative $s_i$ covers agents in a ball of radius $\delta_i\leq \delta_k$. Thus, if there are $\ell$ representatives that are not the last representative that cover something in $S$, then we have at least $\ell$ representatives in $\cup_{v \in S} B(v, \delta_k)\subseteq \cup_{v \in S} B(v, D(S))$, and $S$ satisfies the PRF condition. Now suppose that there are only $\ell-1$ representatives that are not the last representative covering something in $S$. Since $|S|=\ell \frac{n}{k}$, this means that there are exactly $\ell$ representatives covering something in $S$, with each covering exactly $\frac{n}{k}$ agents of $S$, and with one of these being the last representative $x_k$. Moreover, this means that for every representative covering something in $S$, its corresponding set $S_i$ (the set of agents it covers) is entirely contained in $S$, i.e., $S_i\subseteq S$. But the $i$'th representative in our algorithm is always chosen from the set $S_i$: this is true for the last representative as well, since the last representative $x_k$ is not necessarily the center of $S_k$, but it is still chosen from among the elements of $S_k$. Therefore, all $\ell$ representatives covering elements of $S$ are in fact members of $S$ themeselves in this case, and so trivially $\cup_{v \in S} B(v, 0)\subseteq \cup_{v \in S} B(v, D(S))$
    contains at least $\ell$ representatives. This completes the proof that $S$ always satisfies the PRF condition, with $S$ being arbitrary.    
\end{proof}

Now we prove that the algorithm outputs a NORP solution. Again, the argument works similarly as before, except we now need to show that moving the last representative cannot create a situation that causes over representation. Note that the following argument does rely on the fact that the radius of the balls is increasing over time; this is in contrast to the proof of NORP for Algorithm \ref{alg:alg2} where the radius of the chosen balls is not monotone with time.
\begin{lem}
    The solution $X$ of Algorithm \ref{alg:alg1} satisfies NORP. \label{lem:alg1norm}
\end{lem}

\begin{proof}
    Let $S \subseteq X$ be a set of representatives of size $\ell>1$ (or else NORP is trivially satisfied). Let $x_l \in S$ be the representative added the latest in $S$ by Algorithm \ref{alg:alg1}, and likewise $x_s \in S$ be the second latest representative in $S$ selected by Algorithm \ref{alg:alg1}. We will also define $S' = S \setminus \{x_l\}$ to be the set of representatives in $S$, excluding the representative added the latest in $S$. Let $\delta_s$ be the time when $x_s$ was selected. Now we consider the diameter of this set $S$. First, we see that the last center, $x_l$, must be at least $\delta_s$ from $x_s$ since we cover (and thus will never select as representatives) all agents strictly closer than $\delta_s$ to $x_s$ when we picked it (since we cover agents greedily). Thus, we can conclude that the diameter $D(S) \geq \delta_s$.
    
    Now we need to show that there are enough agents near $S$ to justify having $\ell$ representatives. We know that $\delta_j \leq \delta_s$ for all $x_j \in S'$, since $x_j$ was added earlier in the algorithm than $x_s$. Recall that each ball centered on representative $x_j \in S'$ covered $\frac{n}{k}$ agents, all of which were at most $\delta_j$ away from $x_j$. In addition, no agent was removed by two different centers, so for all $x_j \in S'$ the sets they cover are disjoint. 
    So, there are at least $(\ell-1)\frac{n}{k}$ agents in $\cup_{x \in S'} B(x, \delta_s)$, and thus in $\cup_{x \in S'} B(x, D(S))$. Now consider the ball of radius $D(S)$ centered on the last representative, $x_l$. There is at least one additional agent that is not covered before time $\delta_\ell$ at $x_l$, since we chose $x_l$ as a representative and we only choose uncovered agents as representatives. Thus, we have at least $(\ell-1)\frac{n}{k}+1$ agents in $\cup_{x \in S} B(x, D(S))$ for every $S$, so $X$ satisfies NORP. 
\end{proof}

\begin{lem}
    The solution $X$ of Algorithm \ref{alg:alg1} satisfies mJR. \label{lem:alg1mjr}
\end{lem}
\begin{proof}
    Let $S \subseteq N$ be a set of agents of size $\frac{n}{k}$ with radius $R(S)$. Like in a previous lemma, we will consider two cases: $R(S) < \delta_k$ and $R(S) \geq \delta_k$. Let $c$ be the agent in $N$ such that $B(c,R(S)) \cap S \supseteq S$. If $c \in X$, then $X$ meets the properties of mJR immediately. We may therefore proceed assuming that $c \not\in X$.
    
    First, suppose $R(S) < \delta_k$. Then we would have grown the ball around $c$, and if there were $\frac{n}{k}$ uncovered agents at time $R(S)$, we would have added the representative $c$. Otherwise, there was some agent in $S$ that was already covered at time $R(S)$. But since $R(S) < \delta_k$, any agent covered by time $R(S)$ can be at most $R(S)$ away from its covering representative, fulfilling mJR.
    
    Now suppose that $R(S) \geq \delta_k$. At the end of the algorithm, all agents are covered. Again we can assume that we did not pick $c$ or else mJR would immediately be satisfied. If something in $S$ was covered by a set other than $S_k$ then we are done since the representatives in $S_i$ for $i < k$ are at most $\delta_i\leq \delta_k\leq R(S)$ away from all agents in $S_i$. Finally, if all agents of $S$ are covered by $S_k$, then $S=S_k$. Since $X$ contains a representative in $S_k$, we know that mJR is satisfied.
\end{proof}
\begin{thm}
    The solution $X$ of Algorithm \ref{alg:alg1}has $cost(X)\leq 4\cdot cost(OPT)$.\label{thm:four}
\end{thm}
\begin{proof}
    First, we will write the cost of $X$ and apply the triangle inequality simultaneously to the distances $d(v, x_i)$ by going through all $o \in \OPT$ and averaging the distances.
    \begin{align}
        \cost(X) &= \sum_{x_i \in X} \sum_{v \in N} d(v, x_i)\\
        &\leq  \frac{1}{k} \sum_{x_i \in X} \sum_{v \in N} \sum_{o \in \OPT} \left[ d(v, o) + d(o, x_i) \right]\\
        &\leq \sum_{v \in N} \sum_{o \in \OPT}  d(v, o) + \frac{1}{k} \sum_{x_i \in X} \sum_{v \in N} \sum_{o \in \OPT} d(o, x_i)\\
        &\leq \cost(\OPT) + \frac{n}{k} \sum_{x_i \in X}\sum_{o \in \OPT} d(o, x_i)
    \end{align}

    We will now focus on bounding the last term by $3\cdot \cost(\OPT)$.  First we must pull out the the last representative, $x_k$. Let $X' = \{x_i \in X \mid i \not = k\}$. Now we rewrite $\frac{n}{k} \sum_{x_i \in X}\sum_{o \in \OPT} d(o, x_i)$ as 
    \begin{align}
        \frac{n}{k} \sum_{x_i \in X'}\sum_{o \in \OPT} d(o, x_i) + \frac{n}{k}\sum_{o \in \OPT} d(o, x_k). \label{costsplit}
    \end{align}
    We separate $X'$ from $x_k$ since in the algorithm we specifically selected $x_k$ to minimize the distance to $\OPT$ of all agents in $S_k$. Because of this, we know that $\sum_{o \in \OPT} d(o, x_k)\leq \sum_{o \in \OPT} d(o, v)$ for any $v\in S_k$. 
    To bound the first term above, we will apply the triangle inequality by going through the average of all agents in $S_i$, the set of agents covered by $x_i$:
    
    
    \begin{align}
        \frac{n}{k} \sum_{x_i \in X}\sum_{o \in \OPT} d(o, x_i)
        &=\frac{n}{k} \sum_{x_i \in X'}\sum_{o \in \OPT} d(o, x_i) + \frac{n}{k}\sum_{o \in \OPT} d(o, x_k)\\
        &\leq \frac{n}{k} \frac{k}{n} \sum_{x_i \in X'} \sum_{o \in \OPT} \sum_{v \in S_{i}} \left[ d(v, x_i) +  d(v,o)\right] + \sum_{v\in S_k}\sum_{o \in \OPT} d(v,o)\\
        &=\sum_{x_i \in X'} \sum_{o \in \OPT} \sum_{v \in S_{i}} \left[ d(v, x_i) +  d(v,o)\right] + \sum_{v\in S_k}\sum_{o \in \OPT} d(v,o)\\
        &=\sum_{x_i \in X'} \sum_{o \in \OPT}  \sum_{v \in S_{i}}  d(v, x_i) + \sum_{o \in \OPT} \sum_{v\in N} d(v,o)\\
        & =\sum_{x_i \in X'} \sum_{o \in \OPT}  \sum_{v \in S_{i}}  d(v, x_i) + \cost(\OPT)
    \end{align}

    So far, we have shown that $\cost(X)\leq 2\cdot\cost(\OPT) + \sum_{x_i \in X'} \sum_{o \in \OPT} \sum_{v \in S_i}d(v,x_i).$
    Thus, we can focus on the bounding of the term $\sum_{x_i \in X'} \sum_{o \in \OPT} \sum_{v \in S_i}d(v,x_i)$.  We do this by bounding $\sum_{x_i \in X'} \sum_{v \in S_i}d(v,x_i)$ for each $o \in \OPT$ independently.
    We will show how to bound this term in the following key lemma. Note that we are bounding the distances from representatives in $X'$ to agents $S_i$ that they cover (so this does not include the last representative $x_k$), by the distances from $o$ to {\em all agents}, including those covered by $x_k$. Once again, it is crucial that $x_k$ is treated differently, as the following lemma does not hold if $X'$ on the left side is replaced by $X$. 

    \begin{lem}
        For any $o \in \OPT$ we have that $$\sum_{x_i \in X'} \sum_{v \in S_i} d(v, x_i) \leq 2 \sum_{x_i \in X} \sum_{v \in S_i} d(v, o) = 2\sum_{v\in N}d(v,o).$$ \label{lem:charged costs}
    \end{lem}

    \begin{proof}
         We will proceed via a charging argument. The terms $d(v,x_i)$ will be charged to $d(w,o)$, for some $w$ in $S_j$ for $j \geq i$.
    
        First, let us define $N^- = \{v: \exists j < k, v \in S_j\}$, the set of agents covered by representatives not in $S_k$. We will create an injective mapping $g_o$ from $N^-$ to $N$ with the condition that for any $v \in N^-$, if $v$ is covered by $x_i$, then $d(v, x_i) \leq 2d(g(v), o)$. Formally, we introduce the following proposition:
        \begin{prop} \label{prop:mapping}
            For each $o \in \OPT$, there exists an injective mapping $g_o:N^- \to N$ such that for any $v \in N^-$ and $x_i$ covering $v$, we have that $d(v, x_i) \leq 2 d(g_o(v), o)$.
        \end{prop}
        We will reserve the details on how to construct this injective function for the end of the proof. For now we will proceed assuming that the desired injective function $g_o$ exists. We can then conclude that the following bounds for $\sum_{x_i \in X'} \sum_{v \in S_i} d(v, x_i)$ hold.
        \begin{align*}
            \sum_{x_i \in X'} \sum_{v \in S_i} d(v, x_i) &\leq  2\left[\sum_{x_i \in X'} \sum_{v \in S_i} d(g_o(v), o)\right]\\
            &\leq 2\sum_{x_i \in X} \sum_{v \in S_i} d(v, o)
        \end{align*} as desired.
    \end{proof}


    With the above lemma, finishing the proof of Theorem \ref{thm:four} is simple:
    \begin{align*}
        \cost(X) 
        &\leq 2\cost(\OPT) + \sum_{x_i \in X'} \sum_{o \in \OPT} \sum_{v \in S_{i}} d(v, x_i) \\
        &\leq 2\cost(\OPT) + 2\sum_{o \in \OPT} \sum_{v \in N} d(v, o) \\
        &\leq 4 \cost(\OPT).
    \end{align*}
    Now to finish the main technical contribution of this proof, all we have left is to prove the existence of an injective mapping $g_o$ given by Proposition \ref{prop:mapping}. We do this below.
    
    \begin{proof}[Proof of Proposition~\ref{prop:mapping}]
        At a high level we will build the mapping $g_o$ by looking at each ball around $x_i$ of radius $\delta_i$, that consists of agents in the set $S_i$, in the opposite order in which they were added by the algorithm. We will consider two types of agents:``good" agents who prefer their covering center $x_i$ by at least a factor of $2$, i.e.,  $d(v,x_i) \leq 2 d(v,o)$; and ``bad" agents who do not, i.e., $d(v,x_i) > 2 d(v,o)$. Then we will show that at time $\delta_i$, when we choose $x_i$, there cannot be too many agents in the open ball of radius $\frac{1}{2}\delta_i$ around any $o \in \OPT$. This allows us to map ``good" agents to themselves, and ``bad" agents to other agents in a later ball that are sufficiently far away from $o$. Lastly, through a counting argument, we show that we always have enough available slots for the ``bad" agents to charge to.

        The first step in the proof is to show that the open ball $B^\circ(o, \frac{1}{2}\delta_i)$ has few uncovered agents, as follows.
        

        \begin{lem} \label{lem:counting}
            At time $\delta_i$ there are strictly fewer than $\frac{n}{k}$ uncovered agents in $B^\circ(o, \frac{1}{2}\delta_i)$.
        \end{lem}
        \begin{proof}
            Suppose to the contary that there exist at least $\frac{n}{k}$ uncovered agents in $B^\circ(o, \frac{1}{2}\delta_i)$. Then there exists an uncovered $v \in B^\circ(o, \frac{1}{2}\delta_i)$ such that there is a ball of radius strictly less than $\delta_i$ that contains $\frac{n}{k}$ uncovered agents  (in fact this holds for all uncovered $v \in B^\circ(o, \frac{1}{2}\delta_i)$). But we should have selected the ball around $v$ earlier in the algorithm, yielding a contradiction. Note that this holds even if $v$ would have been the last ball, as it is not necessary that we make $v$ a representative, only that we choose $v$'s ball.
        \end{proof}

        Continuing with our proof, we will give an algorithm to construct the injective mapping from $N^- \to N$ with our desired property. To do so, we will consider the representatives in reverse order. 
        Now we will define two sets of agents with respect to each $x_i \in X'$ and to the $o \in \OPT$ that we are considering.
        $$
            S^-_{(i,o)} = \{v \in S_i \mid 2d(v, o) < d(v,x_i)\}
        $$
        and 
        $$
            S^+_{(i,o)} = \{v \in S_i \mid d(v,x_i) \leq 2d(v, o)\}.
        $$
        We can then immediately assign $g_o(v) = v$ for all $v \in S_{(i,o)}^+$ since we see that $d(v, x_i) \leq 2d(v,o)$ so the desired property that $d(v, x_i) \leq 2d(v,o) = 2d(g_o(v), o)$ holds for these $v\in S^+_{(i,o)}$. These are the ``good" agents described earlier in the high level overview.
        
        Now consider the agents $v \in S_{(i,o)}^-$. Notice that these agents in $S^-_{(i,o)}$ are less than $\delta_i$ from $x_i$. We then can conclude that since $2d(v, o) < d(v,x_i) \leq \delta_i$, we then have that $d(v, o) < \frac{1}{2}\delta_i$, so these agents are included in the open ball $B^\circ(o, \frac{1}{2}\delta_i)$ we considered earlier. The set of agents in $S_{(i,o)}^-$ are the aforementioned ``bad" agents that we described earlier. We will need to charge these agents to later uncovered agents. We will now give an iterative algorithm that can construct the remaining values of $g_o$, starting by forming values $g_o(v)$ for $v\in S_{k-1}$ and proceeding in reverse order.

        While constructing $g_o$ we will also form useful sets $A_i$, which we initialize to $A_{k-1} = S_k$. Each set $A_i$ will have the following three properties:
        \begin{enumerate}
            \item $\vert A_i \vert = \frac{n}{k}$;
            \item $A_i \subseteq S_{i+1}\cup\ldots\cup S_k$;
            \item For all $j > i$ there does not exist a $v \in S_j$ such that $g_o(v) \in A_i$.
        \end{enumerate}

        Note that these properties hold for $A_{k-1} = S_{k}$ since $\vert S_k \vert = \frac{n}{k}$, and the third property holds since $g_o$ is only defined on the domain $N^-$.

        Now consider building the mapping $g_o$ for $S_i$, for some $1\leq i \leq k-1$. We will proceed inductively, so suppose that so far we have computed the mapping for all $v \in S_j$ for $j > i$, and we have constructed the set $A_i$ satisfying the properties above. 

        As described previously, for $v \in S_{(i,o)}^+$, we simply set $g_o(v)=v$. Let $m=|S_{(i,o)}^-|$: this is how many agents of $S_i$ we still need to assign using $g_o$. 
        Now consider the uncovered agents at time $\delta_i$ in Algorithm \ref{alg:alg1}, right before we chose $x_i$. Both $S_i$ and $A_i$ are uncovered, by property 2 of $A_i$ above, since at time $\delta_i$ only $S_1\cup\ldots\cup S_{i-1}$ have been covered. As we argued above, $S_{(i,o)}^-\subseteq B^\circ(o, \frac{1}{2}\delta_i)$, but Lemma \ref{lem:counting} states that there can only be strictly fewer than $\frac{n}{k}$ uncovered agents in this ball. Thus, since $m$ of the agents in $S_{i}$ are also in $B^\circ(o, \frac{1}{2}\delta_{i})$, there must be at most $\frac{n}{k}-1-m$ agents $v \in A_i$ inside this ball, and thus there can be no fewer than $\frac{n}{k}-(\frac{n}{k}-1-m) = m+1$ agents $v\in A_i$ with $d(v,o) \geq \frac{1}{2}\delta_i$. We can then map the agents of $S^-_{(i,o)}$ each to a unique agent in $A_i$ with $d(v,o) \geq \frac{1}{2}\delta_i$. Now $g_o$ satisfies the desired properties since for $v \in S^-_{(i,o)}$ we have that,
        $$d(v,x_{i}) \leq \delta_{i}\leq 2d(g_o(v), o).
        $$
        We then have that the desired property holds for all $v \in S^-_{(i,o)}$. Note that due to the third property of $A_i$, this mapping is injective, since nothing else has been chosen to map to agents in $A_i$ until now. 
        
        Next, we will show how to construct $A_{i-1}$ to maintain the invariant properties stated above. Let $g_o(S_i)$ be the set of $\frac{n}{k}$ elements which we mapped the agents $v\in S_i$ onto. These are either in $S_i$ (for $v\in S^+_{(i,o)}$) or in $A_i$ (for $v\in S^-_{(i,o)}$). Now set $A_{i-1}=(A_i\cup S_i)\setminus g_o(S_i)$. To complete the proof of this proposition, we need to show that $A_{i-1}$ still has the three invariant properties defined above. $|A_{i-1}|=\frac{n}{k}$, since $S_i$ and $A_i$ are disjont sets of size $\frac{n}{k}$. The second property clearly holds, since $A_{i-1}\subseteq S_i\cup A_i \subseteq S_k \cup S_{k-1} \cup ... \cup S_i$. The third property holds as well, since we made sure that $A_{i-1}$ consists only of agents which $g_o$ has not mapped anything onto yet, and we have already finished forming the mapping for all $S_i\cup\ldots\cup S_{k-1}$. 

        This completes the proof of the proposition, since the above process has been shown inductively to construct an injective mapping with the desired property.
    \end{proof}

    Now that we have shown that $g_o$ exists, that completes our proof of the theorem, and we obtain an approximation bound of 4, as desired.
    \end{proof}

    \section{Approximate Fairness} \label{sec:approx}
    In previous sections, we gave $2$-approximation algorithms which achieve PRF or NORP separately, as well as a $4$-approximation algorithm which achieves PRF, NORP, and mJR simultaneously. In this section, we will investigate what happens if we relax the fairness notions and consider approximate notions of fairness. We first show that allowing a 2-PRF solution, we can find an 2-approximate solution in polynomial time. We also give a version of Algorithm \ref{alg:alg1} that computes an approximately PRF solution (instead of 1-PRF), but with a lower cost.
    
    Our first approximately fair algorithm starts by growing $n$ balls, each ball centered on a unique agent in $N$. Then, when a ball centered on an agent $v \in N$ has $\frac{n}{k}$ uncovered agents, we greedily cover all agents in the ball (from the center out). Lastly, we choose the agent $v' \in S_i$, that is, the set of agents covered by $v$, that minimizes the distance to all the points in $\OPT$. Notice that we never stop growing a ball and can select the same ball multiple times, but the covered sets are disjoint. Thus, we will never select the same agent twice to be a representative (since an agent can never be covered twice). This algorithm can be seen as a polynomial-time approximate version of the Algorithm \ref{alg:alg2} that finds sets that are at most $2$ times the diameter of the smallest set. Meaning that although we have a solution that is two approximation to $\OPT$ we suffer by finding a solution that is only at worst a two approximate PRF solution. However, we do not suffer the same loss for mJR and find a $2$ approximation to mJR (as opposed to a $4$ approximation as implied by Lemma \ref{lem:prftomJR}).
\begin{algorithm}
\caption{Computes a 2-PRF, 2-mJR, 2-approx to OPT in poly-time }\label{alg:alg3}
\begin{algorithmic}[1]
\State $Z \gets N$ \Comment{Uncovered agents}
\State $\delta \gets 0$
\State $i \gets 1$
\While{$Z \neq \emptyset$}
\While{$\exists x \in N : \vert B(x, \delta) \cap Z\vert \geq \frac{n}{k}$} \Comment{Select a ball centered at }
\State $S_i \gets \arg \min_{\substack{T \subseteq B(x, \delta) \cap Z \\ \vert T \vert = \frac{n}{k}}} \sum_{v \in T} d(v, x)$
\State $Z \gets Z \setminus S_i$
\State $s_i \gets x$
\State $\delta_i \gets \delta$
\State $i \gets i + 1$
\EndWhile
\State $\text{smoothly increase } \delta$
\EndWhile
\State $X \gets \{ \arg \min_{v \in S_i} \sum_{o \in \OPT} d(v, o) \mid i \in [k] \}$
\end{algorithmic}
\end{algorithm}

\begin{thm} \label{thm:22inf}
    Let $X$ be the solution returned by Algorithm \ref{alg:alg3}. Then $X$ satisfies 2-PRF and 2-mJR, and $\cost(X)\leq 2\cost(\OPT)$.
\end{thm}

\begin{lem}
    Let $X$ be the solution returned by Algorithm \ref{alg:alg3}. Then $X$ is 2-PRF.
\end{lem}

\begin{proof}
    Let $S \subseteq N$ be a set of agents of size $\ell \frac{n}{k}$ with diameter $D(S)$. After time $D(S)$ there must be strictly fewer than $\frac{n}{k}$ uncovered agents in $S$. Since each representative covers at most $\frac{n}{k}$ agents and no agent is covered by the same representative twice we have that there must be $\ell$ distinct representatives covering something in $S$ by time $D(S)$. 

    Recall that the algorithm works by choosing a ball with center $s_i$ and radius $\delta_i$ containing a set $S_i$, and then choosing a representative $x_i$ to be the closest element of $S_i$ to $\OPT$. Let $x_i$ be a representative covering some $v\in S$ by time $D(S)$, and $s_i$ be the center of the set $S_i$. Then we have that $d(v,x_i)\leq d(v,s_i)+d(s_i,x_i) \leq 2\delta_i\leq 2D(S)$. The last inequality is because $x_i$ was placed before time $D(S)$ and thus $\delta_i\leq D(S)$. Thus there are at least $\ell$ representatives in $\cup_{v \in S} B(v, 2\cdot D(S))$.
\end{proof}

\begin{lem}
    Let $X$ be the set of representatives returned by Algorithm \ref{alg:alg3}. Then the clustering $X$ has cost $\cost(X) \leq 2 \cdot \cost(\OPT)$.
\end{lem}

\begin{proof} Follows from Proposition \ref{prop:2approx}.
\end{proof}

\begin{lem}
     Let $X$ be the solution returned by Algorithm \ref{alg:alg3}. Then $X$ is 2-mJR.
\end{lem}

\begin{proof}
    Let $S \subseteq N$ be a subset of $N$ of size $\frac{n}{k}$ and let $S$ have radius $R(S)$. Let $c \in N$ be the agent in $N$ such that $S \subseteq B(c, R(S))$ (recall that $c$ need not be in $S$). 
    Note that the set $S$ must have something covered after time $R(S)$ or else there is a ball of radius $R(S)$ centered at $c$ that contains $S$ that would cover all of $S$.
    
    Now consider how far an agent $v$ can be from some $x_i$ that covers it by time $R(S)$. Let $s_i$ be the center of the ball with radius $\delta_i\leq R(S)$ chosen when forming the set $S_i$. Then, $d(v, x_i) \leq d(v,s_i)+d(s_i,x_i)\leq 2\delta_i\leq 2R(S)$. Thus, there must be a representative at most $2R(S)$ away from an agent in $S$. 
\end{proof}
    
    Our final algorithm is parameterized by $\alpha \geq 1$ which is the PRF approximation factor. The algorithm works identically to Algorithm \ref{alg:alg1} with the exception that it maintains a special agent whose ball is grown $\alpha$ times faster than the rest. The special agent is the one who is closest to $o_1 = \arg \min _{o \in \OPT} \sum_{v \in N}d(v, o)$. By growing the balls near $o_1$ $\alpha$ times faster we can improve the approximation bounds to $2 + \frac{2}{\alpha}$, with the PRF approximation growing to $\alpha$ and all while maintaining NORP.

\begin{algorithm}
\caption{($2 + \frac{2}{\alpha}$,$\alpha$,1) algorithm}\label{alg:alg5}
\begin{algorithmic}[1]
\State $Z \gets N$
\State $\delta \gets 0$
\State $i \gets 1$
\State $a \gets\arg \min_{v \in Z} d(v, o_1)$ \Comment{Break ties in the $\arg \min$ arbitrarily}
\While{$Z \neq \emptyset$}
\While{$\exists x \in Z: (\vert B(x, \delta) \cap Z\vert \geq \frac{n}{k}) \lor ((x = a)\land (\vert B(x, \alpha \cdot \delta) \cap Z\vert \geq \frac{n}{k} ))$}
\State $S_i \gets \arg \min_{\substack{T \subseteq Z \\ \vert T \vert = \frac{n}{k}}} \sum_{v \in T} d(v, x)$
\State $Z \gets Z \setminus S_i$
\State $s_i \gets x$
\State $\delta_i \gets \delta$
\State $i \gets i + 1$
\State $a \gets\arg \min_{v \in Z} d(v, o_1)$
\EndWhile
\State $\text{smoothly increase } \delta$
\EndWhile
\State $X \gets \{s_i \vert i < k\} \cup \{ \arg \min_{v \in S_k} d(v, o_1)\}$
\end{algorithmic}

\end{algorithm}

\begin{lem}
    Let $X$ be the solution returned by Algorithm \ref{alg:alg5}. Then $X$ is $\alpha$-PRF.
\end{lem}
\begin{proof}
    The proof is identical to the proof of Lemma \ref{lem:alg1prf}, with the exception that the balls now are at most $\alpha \delta$ in radius.
\end{proof}
\begin{lem}
    Let $X$ be the solution returned by Algorithm \ref{alg:alg5}. Then $X$ is $\alpha$-mJR.
\end{lem}
\begin{proof}
    The proof is identical to the proof of Lemma \ref{lem:alg1mjr}, with the exception that the balls now are at most $\alpha \delta$ in radius and we consider the two cases $R(S) \leq \alpha \delta_k$ and $R(S) > \alpha \delta_k$.
\end{proof}
\begin{lem}
Let $X$ be the solution returned by Algorithm \ref{alg:alg5}. Then $X$ satisfies NORP.\end{lem}
\begin{proof}
    The proof is almost identical to the proof of Lemma \ref{lem:alg1norm} since we will never select a representative that is within each of the previous representative's ball. The only difference is that instead of setting $x_s$ to be the second-to-last representative selected, we set $x_s$ to be the representative with the largest radius of the set it covers (either $\delta_i$ or $\alpha\delta_i$ depending on how $x_i$ was selected). 
\end{proof}

\begin{thm}
    Algorithm \ref{alg:alg5} outputs a solution $X$ with $\cost(X) \leq (2 + \frac{2}{\alpha})\cost(\OPT)$.
\end{thm}

\begin{proof}
    First, we will write the cost of $X$ and apply the triangle inequality to the distances $d(v, x_i)$ by going through $o_1 \in \OPT$ such that $o_1 \in \arg \min_{o \in N} \sum_{v \in N}d(v, o)$.
    \begin{align}
        \cost(X) &= \sum_{x_i \in X} \sum_{v \in N} d(v, x_i)\\
        &\leq  \sum_{x_i \in X} \sum_{v \in N} \left[ d(v, o_1) + d(o_1, x_i) \right]\\
        &\leq k \sum_{v \in N} d(v, o_1) + n \sum_{x_i \in X} d(o_1, x_i)\\
        &\leq \sum_{v \in N} \sum_{o \in \OPT}  d(v, o) + n \sum_{x_i \in X} d(o_1, x_i)\\
        &\leq \cost(\OPT) + n \sum_{x_i \in X} d(o_1, x_i)
    \end{align}

    We will now focus on bounding the last term by $(1 + \frac{2}{\alpha})\cdot \cost(\OPT)$.  First we must pull out the the last representative, $x_k$. Let $X' = \{x_i \in X \mid i \not = k\}$. Now we rewrite $n \sum_{x_i \in X}d(o_1, x_i)$ as 
    \begin{align}
        n \sum_{x_i \in X'}d(o_1, x_i) + n d(o_1, x_k). 
    \end{align}
    We separate $X'$ from $x_k$ since in the algorithm we specifically selected $x_k$ to minimize the distance to $o_1$ of all agents in $S_k$. Because of this, we know that $d(o_1, x_k)\leq d(o_1, v)$ for any $v\in S_k$. 
    
    
    Once again we need to split $X'$ into two sets, those who were selected when chosen as $a$ in the algorithm (excluding $x_k)$ and those who were not. We call the two aforementioned sets $Y$ and $Y'$ (resp). $Y$ is the set of representatives that were chosen with radius $\alpha\delta$ and were the closest uncovered agent to $o_1$ at that time. $Y'$ is the rest of representatives chosen in the ``normal" way, because the ball of radius $\delta$ centered at that agent contained $\frac{n}{k}$ uncovered agents. In the inequalities below, we use the triangle inequality, and the fact that for any $x_i\in Y$, $v\in S_i$, it must be that $d(x_i,o_1)\leq d(v,o_1)$, since $x_i$ was the closest to $o_1$ at the time it was chosen.
    \begin{align}
        n \sum_{x_i \in X'}d(o_1, x_i) + n d(o_1, x_k) &= n \sum_{x_i \in Y}d(o_1, x_i) + n \sum_{x_i \in Y'}d(o_1, x_i) + n d(o_1, x_k)
        \\&\leq n \sum_{x_i \in Y}d(o_1, x_i)  + n \frac{k}{n} \sum_{x_i \in Y'} \sum_{v \in S_{i}} \left[ d(v, x_i) +  d(v,o_1)\right] + k\sum_{v\in S_k} d(v,o_1)\\
        &=n \sum_{x_i \in Y}d(o_1, x_i)  + k\sum_{x_i \in Y'} \sum_{v \in S_{i}} \left[ d(v, x_i) +  d(v,o_1)\right] + k\sum_{v\in S_k} d(v,o_1)\\
        &\leq k \sum_{x_i \in Y} \sum_{v \in S_i}d(v, o_1)  + k\sum_{x_i \in Y'}  \sum_{v \in S_{i}}  \left[ d(v, x_i) +  d(v,o_1)\right] + k\sum_{v\in S_k} d(v,o_1)\\
        &\leq \cost(\OPT) + k \sum_{x_i \in Y'} \sum_{v \in S_i} d(v,x_i)
    \end{align}

    Let us now focus on the bounding of the term $k \sum_{x_i \in Y'} \sum_{v \in S_i}d(v,x_i)$. 
    We will show how to bound this term in the following key lemma. Note that we are bounding the distances from representatives in $Y'$ to agents $S_i$ that they cover (so this does not include the last representative $x_k$), by the distances from $o_1$ to {\em all agents }, including those covered by $x_k$. 

    \begin{lem}
        We have that $$\sum_{x_i \in Y'} \sum_{v \in S_i} d(v, x_i) \leq \frac{2}{\alpha} \sum_{x_i \in Y' \cup \{x_k\}} \sum_{v \in S_i} d(v, o_1).$$ 
    \end{lem}

    \begin{proof}
         We will proceed via a charging argument. The terms $d(v,x_i)$ will be charged to $d(w,o_1)$, for some $w$ in $S_j$ for $j \geq i$ and $x_j \in Y' \cup \{x_k\}$.
    
        First, let us define $N^- = \cup_{x_j \in Y'} S_j$, which is the set of agents covered by representatives other than the representative $a$ at the time it was selected (and not representative $x_k$). Likewise we define $N^+ = N^- \cup \{S_k\}$. We will create an injective mapping $g$ from $N^-$ to $N^+$ with the condition that for any $v \in N^-$, if $v$ is covered by $x_i$, then $d(v, x_i) \leq \frac{2}{\alpha}d(g(v), o_1)$. Formally, we introduce the following proposition defining the mapping we desire to construct.
        \begin{prop} \label{prop:mapping-2}
            There exists an injective mapping $g:N^- \to N^+$ such that for any $v \in N^-$ and $x_i$ covering $v$, we have that $d(v, x_i) \leq \frac{2}{\alpha} d(g(v), o_1)$.
        \end{prop}
        We will reserve the details on how to construct this injective function for the end of the proof. For now we will proceed assuming that the desired injective function $g$ exists. We can then conclude that the following bounds for $\sum_{x_i \in Y'} \sum_{v \in S_i} d(v, x_i)$ hold.
        \begin{align*}
            \sum_{x_i \in Y'} \sum_{v \in S_i} d(v, x_i) &\leq  \frac{2}{\alpha}\left[\sum_{x_i \in Y'} \sum_{v \in S_i} d(g(v), o_1)\right]\\
            &\leq \frac{2}{\alpha}\sum_{x_i \in Y' \cup \{x_k\}} \sum_{v \in S_i} d(v, o_1)
        \end{align*} as desired.
    \end{proof}


    With the above lemma, finishing the proof of Theorem \ref{thm:four} is simple:
    \begin{align*}
        \cost(X) 
        &\leq 2\cost(\OPT) +  k\sum_{x_i \in Y'}  \sum_{v \in S_{i}}  d(v, x_i) \\
        &\leq 2\cost(\OPT) +  k\frac{2}{\alpha}\sum_{x_i \in Y' \cup \{x_k\}} \sum_{v \in S_i} d(v, o_1) \\
        &\leq 2\cost(\OPT) +  k \frac{2}{\alpha}\sum_{x_i \in X} \sum_{v \in S_i}d(v, o_1) \\
        &= 2\cost(\OPT) +  k \frac{2}{\alpha}\sum_{v \in N}d(v, o_1) \\
        &\leq 2\cost(\OPT) +  \frac{2}{\alpha}\sum_{v \in N}\sum_{o \in \OPT}d(v, o) \\
        &\leq \left(2 + \frac{2}{\alpha}\right) \cost(\OPT).
    \end{align*}
    Now to finish the main technical contribution of this proof, all we have left is to prove the existence of an injective mapping $g$ given by Proposition \ref{prop:mapping-2}. We do this below.
    
    \begin{proof}[Proof of Proposition~\ref{prop:mapping-2}]
        At a high level we will build the mapping $g$ by looking at each ball around $x_i$ of radius $\delta_i$, that consists of agents in the set $S_i$, in the opposite order in which they were added by the algorithm. We will consider two types of agents:``good" agents who prefer their covering center $x_i$ by at least a factor of $2/\alpha$, i.e.,  $d(v,x_i) \leq \frac{2}{\alpha} d(v,o_1)$; and ``bad" agents who do not, i.e., $d(v,x_i) > \frac{2}{\alpha} d(v,o_1)$. Then we will show that at time $\delta_i$, when we choose $x_i$, there cannot be too many agents in the open ball of radius $\frac{1}{2}\alpha\delta_i$ around $o_1$. This allows us to map ``good" agents to themselves, and ``bad" agents to other agents in a later ball that are sufficiently far away from $o_1$. Lastly, through a counting argument, we show that we always have enough available slots for the ``bad" agents to charge to.

        The first step in the proof is to show that the open ball $B^\circ(o_1, \frac{1}{2}\alpha\delta_i)$ has few uncovered agents, as follows.
        

        \begin{lem} \label{lem:counting2}
            At time $\delta_i$ there are strictly fewer than $\frac{n}{k}$ uncovered agents in $B^\circ(o_1, \frac{1}{2}\alpha\delta_i)$.
        \end{lem}
        \begin{proof}
            Suppose to the contrary that there exist at least $\frac{n}{k}$ uncovered agents in $B^\circ(o_1, \frac{1}{2}\alpha\delta_i)$ at time $\delta_i$. Then the uncovered agent $a$ closest to $o_1$ will have all of these agents contained in the ball $B^\circ(a, \alpha\delta_i)$, and so these agents will become covered by our algorithm before time $\delta_i$, giving a contradiction.
        \end{proof}

        Continuing with our proof, we will give an algorithm to construct the injective mapping from $N^- \to N^+$ with our desired property. To do so, we will consider the representatives in reverse order. 
        Now we will define two sets of agents with respect to each $x_i \in Y'$.
        $$
            S^-_i = \left\{v \in S_i \mid \frac{2}{\alpha}d(v, o_1) < d(v,x_i)\right\}
        $$
        and 
        $$
            S^+_i = \left\{v \in S_i \mid d(v,x_i) \leq \frac{2}{\alpha}d(v, o_1)\right\}.
        $$
        We can then immediately assign $g(v) = v$ for all $v \in S_i^+$ since we see that $d(v, x_i) \leq \frac{2}{\alpha}d(v,o)$ so the desired property that $d(v, x_i) \leq \frac{2}{\alpha}d(v,o) = \frac{2}{\alpha}d(g(v), o)$ holds for these $v\in S^+_i$. These are the ``good" agents described earlier in the high level overview.

        Now consider the agents $v \in S_i^-$. Notice that these agents in $S^-_{(i,o)}$ are less than $\delta_i$ from $x_i$. We then can conclude that since $\frac{2}{\alpha}d(v, o_1) < d(v,x_i) \leq \delta_i$, we then have that $d(v, o_1) < \alpha\frac{1}{2}\delta_i$, so these agents are included in the open ball $B^\circ(o, \frac{1}{2}\delta_i)$ we considered earlier. The set of agents in $S_i^-$ are the aforementioned ``bad" agents that we described earlier. We will need to charge these agents to later uncovered agents. We will now give an iterative algorithm that can construct the remaining values of $g$, starting by forming values $g(v)$ for $v\in S_{k-1}$ and proceeding in reverse order. First, we will relabel the sets $S_i$ and the representatives $x_i$. Let $k' = \vert Y'\vert + 1$, which is the number of representatives not equal to $a$ when selected, as well as the last representative. Let $y_i$ be the $i$th representative added in $Y'$ and $S'_i$ be its respective covering set, with $y_{k'}=x_k$. Notice that $g$ then takes agents in $N^- = S'_1 \cup S'_2 \cup ...\cup S'_{k'-1}$ and maps them to agents in $N^+ = S'_1 \cup S'_2 \cup ...\cup S'_{k'}$.

        While constructing $g$ we will also form useful sets $A_i$, which we initialize to $A_{k'-1} = S'_{k'}$. Each set $A_i$ with have the following three properties:
        \begin{enumerate}
            \item $\vert A_i \vert = \frac{n}{k}$;
            \item $A_i \subseteq S'_{i+1}\cup\ldots\cup S'_{k'}$;
            \item For all $j > i$ there does not exist a $v \in S'_j$ such that $g(v) \in A_i$.
        \end{enumerate}

        Note that these properties hold for $A_{k'-1} = S'_{k'}$ since $\vert S'_{k'} \vert = \frac{n}{k}$, the second property holds since all sets $S'_i$ are disjoint, and the third property holds since the domain of $g$ does not include ${S'}_{k'}$.

       Now consider building the mapping $g$ for $S'_i$, for some $1\leq i \leq k'-1$. We will proceed inductively, so suppose that so far we have computed the mapping for all $v \in S'_j$ for $j > i$, and we have constructed the set $A_i$ satisfying the properties above. 

        As described previously, for $v \in {S'}_{i}^+$, we simply set $g(v)=v$. Let $m=|{S'}_{i}^-|$: this is how many agents of $S'_i$ we still need to assign using $g$. 
        Now consider the uncovered agents at time $\delta_i$ in Algorithm \ref{alg:alg5}, right before we choose $y_i$. Both $S'_i$ and $A'_i$ are uncovered at this time, by property 2 of $A_i$ above. As we argued above, ${S'}_{i}^-\subseteq B^\circ(o_1, \frac{1}{2}\delta_i)$, but Lemma \ref{lem:counting2} states that there can only be strictly fewer than $\frac{n}{k}$ uncovered agents in this ball. 
        Thus, since $m$ of the agents in ${S'}_{i}$ are also in $B^\circ(o, \frac{1}{2}\alpha\delta_{i})$,
        there must be at most $\frac{n}{k}-1-m$ agents $v \in A_i$ inside this ball, and thus there can be no fewer than 
        $\frac{n}{k}-(\frac{n}{k}-1-m) = m+1$ agents $v\in A_i$ with $d(v,o_1) \geq \frac{\alpha}{2}\delta_i$. 
        We can then map the agents of ${S'}^-_{i}$ each to a unique agent in $A_i$ with $d(v,o) \geq \frac{1}{2}\alpha\delta_i$. Now $g$ satisfies the desired properties since for $v \in S^-_{i}$ we have that,
        $$d(v,y_{i}) \leq \delta_{i}\leq \frac{2}{\alpha}d(g(v), o_1).$$
        We then have that the desired property holds for all $v \in {S'}^-_{i}$. Note that due to the third property of $A_i$, this mapping is injective, since nothing else has been chosen to map to agents in $A_i$ until now. 
        
        Next, we will show how to construct $A_{i-1}$ to maintain the invariant properties stated above. Let $g(S_i')$ be the set of $\frac{n}{k}$ elements which we mapped the agents $v\in {S'}_i$ onto. These are either in ${S'}_i$ (for $v\in {S'}^+_{i}$) or in $A_i$ (for $v\in {S'}^-_{i}$). Now set $A_{i-1}=(A_i\cup S_i')\setminus g(S_i')$. To complete the proof of this proposition, we need to show that $A_{i-1}$ still has the three invariant properties defined above. $|A_{i-1}|=\frac{n}{k}$, since $S_i$ and $A_i$ are disjoint sets of size $\frac{n}{k}$. The second property clearly holds, since $A_{i-1}\subseteq {S'}_i\cup A_i \subseteq {S'}_{k'} \cup {S'}_{k'-1} \cup ... \cup {S'}_i$. The third property holds as well, since we made sure that $A_{i-1}$ consists only of agents which $g$ has not mapped anything onto yet, and we have already finished forming the mapping for all $S_i'\cup\ldots\cup S_{k'-1}'$. 

        This completes the proof of the proposition, since the above process has been shown inductively to construct an injective mapping with the desired property.
    \end{proof}

    Now that we have shown that $g$ exists, that completes our proof of the theorem, and we obtain an approximation bound of $2+\frac{2}{\alpha}$, as desired.
    \end{proof}

\section{Conclusion}

In this manuscript, we showed that many known fairness axioms are compatible with good solutions according to the sum-cost objective, i.e., that we can find simultaneously low-cost and fair solutions. We also defined NORP, a new fairness axiom which focuses on avoiding over-representation. We proved that NORP is compatible with many other fairness axioms as well, and that we can find fair and NORP solutions which approximate the optimum for the sum objective within a constant factor. 

There are still many open questions which remain. Can a solution always be found that is PRF and NORP while also better than a four approximation to $\OPT$? 
Can a 2-approximate and PRF solution always be found in polynomial time for a general metric space?
Finally, what other objectives are compatible with NORP and fairness axioms, i.e., could we form constant approximations to these objectives while satisfying fairness constraints? One promising class of such objectives is the $\ell$-centrum objectives \cite{slater1978centers,tamir2001k,peeters1998some, han2023optimizing}: can similar results be shown for that class, perhaps with the approximation ratio depending on $\ell$?
\bibliography{refs}

\appendix

\section{Extending our results when $k$ does not divide $n$}
\label{app:kdividesn}

Recall that in the body of the paper, we considered only instances where $k$ divides $n$, and only proved our results about solution costs and fairness under that assumption. We now justify that this assumption is reasonable, and show that all our results in this work extend to the more general case where $k$ may not divide $n$. Rather than form new algorithms for this more general case, we instead give a reduction from the general case where $k$ may not divide $n$ to the special case when it does. Our reduction will show that if an algorithm provides low-cost, or fair, or NORP solutions for the special case, then it can be transformed into an algorithm for the general case with the same properties. 

Before we present our reduction, we note that it only works for algorithms with the {\em point-NORP} property, defined below. 
We will need to consider balls of radius $0$, which are defined as 
$B(v) = \{v' \in N \mid d(v,v') = 0\},$ 
which are the agents from the set $N$ that lie at the same point in the metric space. Note that since multiple agents are allowed to be at the same point (i.e., our space is a pseudo-metric), balls of radius 0 may have many agents in them. 

\begin{defn}[Point-NORP] \label{prop:nottoobunched} Let $X$ be a set of representatives of size $k$. For any set $S \subseteq X$ of size $\vert S \vert \geq \ell$ such that $S\subseteq B(v)$ for some $v$ (i.e., the maximum distance between representatives in $S$ is $0$) the following holds: $ \vert B(v) \vert > (\ell-1)\frac{n}{k}$. 
\end{defn}

Point-NORP is simply a weaker version of NORP applied only to sets $S$ of diameter 0, as it says that any coalition of diameter $0$ cannot have too many representatives. We say that an algorithm satisfies Point-NORP if it always produces committees which satsify Point-NORP. Before we proceed with our reduction showing how to extend results to the case when $k$ does not divide $n$, we first prove that all the algorithms presented in this paper satsify this property. This is trivially true for all algorithms which satsify NORP (since Point-NORP is a special case of NORP for sets of diameter zero), so we just need to show this still holds for Algorithms \ref{alg:alg2} and \ref{alg:alg3}.





\begin{prop}
    Algorithm \ref{alg:alg2} satisfies Point-NORP if ties are broken so that the set with the most agents closest to $\OPT$ is selected first.
\end{prop}
\begin{proof}
 Suppose we have $\ell$ representatives in some $B(v)$. We will argue that at most one representative $x_i$ in $B(v)$ contains agents in $S_i$ outside of $B(v)$. 
 
 So see why this is true, first note that when there is a set of $\frac{n}{k}$ agents of diameter zero, then this set will be chosen first by the algorithm. Thus, at the very beginning of the algorithm (when $\delta=0$), these sets will be selected and covered, with each set $S_i$ contained entirely inside $B(v)$. This continues until less than $\frac{n}{k}$ uncovered agents are left in $B(v)$.
 
 We now claim that at most one more representative will be chosen in $B(v)$. Suppose to the contrary that $x_i$ and $x_j$ in $B(v)$ are both chosen, and without loss of generality let $i < j$. Then at the time of choosing $x_i$, there exist less than $\frac{n}{k}$ uncovered agents $v$ with $d(v,x_i) = 0$, so our tie breaking rule would have kicked in and we would have selected those agents as part of the set $S_i$. Thus we have that at most one representative $x_i$ in $B(v)$ contains agents in $S_i$ outside of $B(v)$.

Since at least $\ell-1$ of the representatives in $B(v)$ have their entire sets $S_i$ contained in $B(v)$, and since each $S_i$ contains exactly $\frac{n}{k}$ agents, then together with the $\ell$'th representative which is also in $B(v)$, it must be that $|B(v)|> (\ell-1)\frac{n}{k}$.
\end{proof}





 \begin{prop}
     Algorithm \ref{alg:alg3} satisfies Point-NORP if we greedily cover $B(x_i)$ first.
 \end{prop} 
 \begin{proof}
 The argument is mostly the same as the previous proposition. By similar reasoning, all representatives in $B(v)$ which are chosen before $B(v)$ had less than $\frac{n}{k}$ uncovered agents will have their entire set $S_i$ in $B(v)$.

 We change the algorithm slightly when $\delta >0$. Instead, we will let $$x_i = \underset{v \in B(x, \delta) \cap Z}{\arg \min} \sum_{o \in \OPT}d(v,o),$$ 
and then we cover the agents $$S_i = (B(x_i) \cap Z) \cup \underset{{T \subseteq B(x,\delta) \cap Z, \vert T \vert = \frac{n}{k} - \vert B(x_i) \cap Z \vert}} {\arg \min} \sum_{v \in T} d(v,x)$$ at each iteration. 
 
 Suppose that we have $\beta$ representatives in $B(v)$, then since we greedily form $S_i$ out from this point, there must be at least $(\beta-1)\frac{n}{k} + 1 \geq (\beta-1) k + 1$ agents in $B(v)$, by exactly the same argument as in the previous proposition.

 Notice that after this change we still return a solution that is a 2-approximation to $\OPT$, since the representative $x_i \in S_i$ still satisfies $\sum_{o \in \OPT}d(o,x_i) \leq \sum_{o \in \OPT} d(o, v)$ for all $v \in S_i$.

 Likewise, the proof of PRF never used the fact that we covered {\em greedily} out of the center agent, only that we choose $S_i$ to be {\em some} set of $\frac{n}{k}$ agents in the ball $B(x,\delta)$.
 \end{proof}



Now that we know all the algorithms in our work satisfy Point-NORP, we proceed with the reduction showing how to extend our results to the case when $k$ does not divide $n$. Suppose we have an Algorithm $A$ which satisfies Point-NORP, and works when $k$ divides $n$. Given a set of agents $N$ of size $n$, and an integer $k$ which does {\em not} divide $n$, we form a new instance with the set of agents $N'= N \times \{1,2,...,k\}$, and $k$ remaining the same. In other words, for every agent $v\in N$, we now have $k$ identical agents in $N'$, which we'll denote as $v_1,v_2,\ldots, v_k$. The distances remain the same, i.e., for any $v,w\in N$, and any $i$ and $j$, we have that $d(v_i,w_j)=d(v,w)$. Note that the number of agents in $N'$ is exactly $n'=kn$, so $k$ divides the number of agents in the new instance. We can therefore run algorithm $A$ on the new instance, and obtain a committee $X'\subseteq N'$ of size $k$. Finally, we set $X\subseteq N$ to be a committee in the original instance as follows. For every $v\in N$, consider $B(v)$ in $N$ and $B(v_1)$ in $N'$. For every representative in $X'$ located in $B(v_1)$, we select an arbitrary agent in $B(v)$ and add it to $X$. This completes our reduction, as we used an algorithm which works on the case when $k$ divides $n$ to form a committee $X$ for the more general setting when $k$ does not divide $n$. 

To show that our reduction above is well-defined, we need to argue that it is actually possible to form a committee $X$ as above, i.e., that there are enough agents in $B(v)$ to select. In other words, we need to show that $|B(v)|\geq |X'\cap B(v_1)|$.  Due to our construction, we know that $|B(v_1)|=k|B(v)|$.  Since Algorithm $A$ satisfies Point-NORP, we also know that the number of representatives in $X'$ chosen in $B(v_1)$ is strictly less than $|B(v_1)|\frac{k}{|N'|}+1$. Thus, $$|X'\cap B(v_1)|\leq |B(v_1)|\frac{k}{|N'|} = k|B(v)|\frac{k}{kn} = \frac{k}{n}|B(v)| \leq |B(v)|,$$
and so our reduction is well-defined. Also we clearly select at most $k$ representatives. For intuition, our reduction essentially simulates a fractional version of $A$ where instead we cover $\frac{1}{k}$ fractions of agents. We now prove that if algorithm $A$ behaves nicely for the case when $k$ divides $n$, then the above reduction also has the same properties for arbitrary $k$ and $n$. 


\begin{thm}
    Let $X$ be the output of the reduction above, $X'$ be the intermediate instance created in the reduction, and $\OPT'$ be the optimal solution of $N'$. We then have the following:
    \begin{enumerate}
        \item If $\cost_{N'}(X') \leq \alpha \cost_{N'}(\OPT')$ then $\cost_N(X) \leq \alpha \cost_N(\OPT)$.
        \item If $X'$ is an $\alpha$-approximate PRF solution, then $X$ is $\alpha$-approximate PRF.
        \item If $X'$ is an $\alpha$-approximate mJR solution, then $X$ is an $\alpha$-approximate mJR.
        \item If $X'$ is an $\alpha$-approximate NORP solution, then $X$ is $\alpha$-approximate NORP.
    \end{enumerate}
\end{thm}
\begin{proof}\text{ } \newline
    \begin{enumerate}
        \item Notice that the cost is given by
            \begin{align}
                \cost_{N}(X) &= \sum_{v \in N} \sum_{x \in X} d(v,x)\\
                &=\frac{k}{k} \sum_{v \in N} \sum_{x \in X} d(v,x)\\
                &=\frac{1}{k} \sum_{v \in N'} \sum_{x \in X'} d(v,x)\\
                &=\frac{1}{k} \cost_{N'}(X')\\
                &\leq \frac{1}{k}\alpha \cost_{N'}(\OPT')
            \end{align}
        Now let us compare the costs of $\OPT'$ and of $\OPT$. Let $O\subseteq N'$ be the solution obtained by taking each $o\in \OPT$, and adding the agent $o_1\in N'$ to $O$. In other words, $O$ is exactly $\OPT$ mapped to $N'$. Then, $\cost_{N'}(O) = k\cost_N(\OPT)$. But since $\OPT'$ is the optimal solution for $N'$, we know that $\cost_{N'}(\OPT')\leq \cost_{N'}(O)$. Thus, we have that    
            \begin{align}
            \cost_{N}(X) &\leq \frac{1}{k}\alpha \cost_{N'}(\OPT')\\
            &\leq \frac{1}{k}\alpha\cost_{N'}(O)\\
            &= \alpha\cost_N(\OPT)
            \end{align}
            as desired.
        \item Suppose $X$ is a solution given by the above reduction. Let $S \subseteq N$ be a set of $\ell \frac{n}{k}$ agents. Then the set $S' \subseteq N'$ corresponding to duplicating each agent $k$ times has $\ell\frac{nk}{k} = \ell n = \ell\frac{n'}{k}$ agents. Thus, balls $\cup_{v \in S'} B(v,\alpha D(S'))$ have at least $\ell$ agents in $X'$. Notice that $D(S') = D(S)$ and so  $\cup_{v \in S} B(v,\alpha D(S))$ has $\ell$ representatives as well.
        
        \item Similarly to the above argument, we consider $S \subseteq N$ with size $\frac{n}{k}$. We see that the set $S' \subseteq N'$ has $\frac{n'}{k}$ agents, so $\cup_{v \in S'} B(v,\alpha R(S'))$ has a representative in $X'$ and thus $\cup_{v \in S} B(v,\alpha R(S))$ has a representative in $X$.
        
        \item Let $S \subseteq X$ be a set of $\ell$ representatives. Then the set $\cup_{x \in S'} B(x,\alpha D(S))$ has $\ell \frac{n'}{k}$ agents in $N'$ and thus there are at least $\ell \frac{n}{k}$ agents in $N$ in $\cup_{x \in S} B(x,\alpha D(S))$.
        \end{enumerate}
\end{proof}

\section{Other fairness concepts}\label{app:mPJR}
We also remark that mPJR as defined by \citet{kellerhals2023proportional} is incompatible with NORP; See Figure \ref{fig:starbadexample}. We believe that this example shows that mPJR is too strong of a property. While an mPJR solution always exists, this example shows that it many not be very ``fair" or robust for our applications, since it concentrates all the representatives in a very small region.

\begin{figure}[h]
    \centering
    \begin{tikzpicture}
        \node[circle,draw,minimum size=1cm] (center) at (0,0) {\large $k-1$};

        \foreach \angle in {30, 90, 150, 210, 270, 330} {
            \node[circle,draw,fill=black,scale=0.6] (p\angle) at (\angle:3) {};
            \draw (center) -- node[midway, above, sloped] {1} (p\angle);
            \node[above] at (p\angle) {1};
        }

        \node at (180:2.8) {\LARGE $\cdots$};
        \node[left] at (180:3.5) {\large $n-k+1$};

    \end{tikzpicture}
    \caption{A instance where mPJR is incompatible with NORP for any approximation factor. Assume $k \leq \sqrt{\frac{n}{2}}$. There are $n-k+1$ agents on the outside spokes, each $1+\epsilon$ from each other, and $k-1$ agents in the middle. Assume there are $\ell\leq k$ representatives on the outer spokes. Then there exists a set $S \subseteq N$ consisting of $n-2k+1 > (k-1)\frac{n}{k}$ agents all on the spokes that are not representatives. This group has $k-1$ agents all at most distance $1$ from all agents in $S$, therefore it does not satisfy mPJR unless at least $k-1$ representatives are placed in the center. But placing $k-1$ representatives in the center clearly violates NORP.}
    \label{fig:starbadexample}
\end{figure}
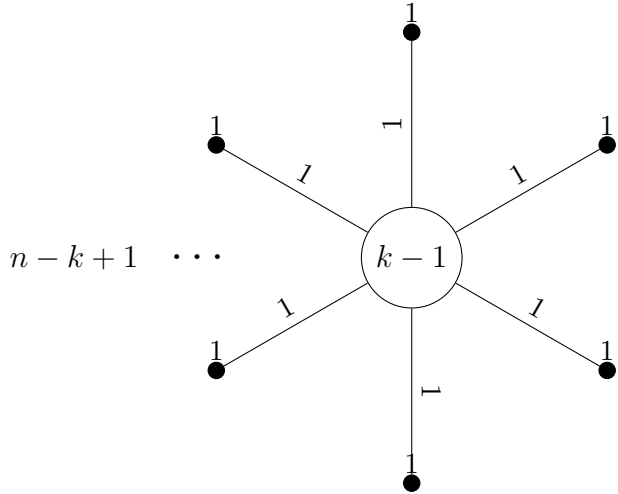

For completeness, we also include the statement of the following theorem, stating that approximate mJR implies many other fairness concepts. \citet{kellerhals2023proportional} already proved these results for exact mJR; we omit the proof of this theorem since it is a direct and simple extension of the proofs in \citet{kellerhals2023proportional}.


\begin{thm}[Theorem 5, \citep{kellerhals2023proportional}]
    \label{thm:kellerhals5}
    Let $X$ be an outcome satisfying $\alpha-mJR$. Then it also satisfies $(1+\sqrt{1+\alpha})$-proportional fairness, $(1+\alpha)$-individual fairness, and is in the $\left(\gamma, \frac{\gamma(\alpha+1)}{\gamma-1}\right)$-transferable core for any $\gamma>1$. 
\end{thm}

\section{Hardness of computing fair and low-cost solutions}
\label{sec:hardness}

\begin{thm}
    For any $\gamma \in (0,1)$ it is co-NP-Hard to determine if there exists a PRF or mJR solution that is better than a $1+\gamma-O\left(\frac{1}{\sqrt{n}}\right)$ approximation to $\OPT$. In particular, it is co-NP-Hard to determine if a PRF or mJR solution that is a $2-O\left(\frac{1}{\sqrt{n}}\right)$ approximation exists. 
\end{thm}

\begin{proof}
    We will show this through a reduction from $CLIQUE$. Note that for any constant $\gamma \in (0,1)$ the problem of determining if a clique in a graph $G=(V,E)$ with size $\vert V \vert\cdot \gamma$ exists is NP-Hard.

    Let $\beta = \gamma  - \frac{1}{\vert V \vert} \in (0,1)$ for $\vert V \vert$ sufficiently large.
    Construct the instance as follows:
    \begin{enumerate}
        \item Create $\vert V \vert$ copies of $G$;
        \item All edges in $G$ have distance $1-\epsilon$;
        \item Create $\frac{1}{\beta} \cdot \vert V \vert$ agents with a distance zero from each other and $1$ from all other agents;
        \item All other distances are the metric completion;
        \item Let $k=\vert V \vert\frac{1}{\beta}$.
    \end{enumerate}
    Now we will show that there exists a clique of size  at least $\gamma \vert V \vert = \beta \vert V \vert + 1$ if and only if there does not exist a PRF or mJR solution that is a $(1+\gamma-O\left(\frac{1}{\vert V \vert}\right))$-approximation to $\OPT$. 
    First observe that for $\vert V \vert > 1$ the solution that minimizes the sum objective places all representatives at the $\vert V \vert\frac{1}{\beta}$ center agents and has cost $\frac{\vert V \vert^3}{\beta}$.
    Note that the number of agents in this instance is $n = \vert V \vert\left(\vert V \vert + \frac{1}{\beta}\right)$ and the quota for a coalition is $\frac{n
    }{k} = \frac{\left(\vert V \vert^2 + \frac{\vert V \vert}{\beta}\right)}{\vert V \vert \frac{1}{\beta}} = \beta \vert V \vert+1 $.
    First, we will show the forward direction.

    Suppose that there exists a clique of size $\gamma \vert V \vert = \frac{n}{k}$. Then to satisfy PRF or mJR, for each copy of $G$ we need to place one of the representatives in the graph instead of on the center agents, since the clique has diameter $1-\epsilon$. Thus, the cost is at least $\vert V \vert^2\left(k-\vert V \vert\right)+2\vert V \vert^2\left(\vert V \vert-1\right)+\vert V \vert\left(\vert V \vert-1\right)(1-\epsilon) = \vert V \vert^2\left(\frac{\vert V \vert}{\beta}-\vert V \vert\right) + 2\vert V \vert^2\left(\vert V \vert-1\right)+\vert V \vert\left(\vert V \vert-1\right)(1-\epsilon)$ (ignoring the cost of the $k$ central agents) while the cost of $\OPT$ is $\vert V \vert^2k = \vert V \vert^3\frac{1}{\beta}$. We then have that the approximation ratio is given by:
    \begin{align*}
        \frac{\vert V \vert^2\left(\frac{\vert V \vert}{\beta}-\vert V \vert\right) + 2n^2\left(\vert V \vert-1\right)+\vert V \vert\left(\vert V \vert-1\right)(1-\epsilon)}{\vert V \vert^3\frac{1}{\beta}} &= 1 + \beta - \frac{\beta}{\vert V \vert} - \frac{\beta}{\vert V \vert^2} - \frac{\beta \epsilon}{\vert V \vert} + \frac{\beta \epsilon}{\vert V \vert^2}\\
        &\geq 1 + \beta - O\left(\frac{1}{\vert V \vert}\right)\\
        &= 1 + \gamma - O\left(\frac{1}{\vert V \vert}\right)\\
        &= 1 + \gamma - O\left(\frac{1}{\sqrt{n}}\right).\\
    \end{align*}
    
    Now for the reverse direction, suppose that there exists a PRF solution that strictly better than a $(1+\gamma-O(\frac{1}{\sqrt{n}}))$- approximation, then by before we know it cannot be the case that each duplicate of the original graph $G$ has a representative, thus there cannot be a clique of size $\gamma \vert V \vert$. This means that for any $\gamma \in (0,1)$ determining if a solution exists that is $1+\gamma - O(\frac{1}{\sqrt n})$ approximate to OPT and is PRF or mJR solution is computationally harder than to determine if a clique of size $\gamma \vert V \vert$ exists.
\end{proof}

\end{document}